\documentclass[a4paper]{article}
\usepackage[utf8x]{inputenc}

\title{Trapezoidal Diagrams, Upward Triangulations, and Prime Catalan Numbers}

\author{Manuel Wettstein\footnote{Department of Computer Science, ETH Z\"urich, Switzerland. E-mail: \url{mw@inf.ethz.ch}}}

\usepackage{mathtools}
\usepackage{amssymb}
\usepackage{amsfonts}

\usepackage{amsthm}

\theoremstyle{plain}
\newtheorem{theorem}{Theorem}
\newtheorem{lemma}[theorem]{Lemma}
\newtheorem{corollary}[theorem]{Corollary}
\newtheorem{observation}[theorem]{Observation}
\newtheorem{definition}[theorem]{Definition}

\usepackage[shortlabels]{enumitem}
\usepackage{array}

\usepackage{rotating,booktabs}

\usepackage{calc}

\usepackage{tikz}
\usetikzlibrary{shapes}
\usetikzlibrary{calc}
\usetikzlibrary{intersections}
\usetikzlibrary{decorations.markings}
\usetikzlibrary{shapes.arrows}

\usepackage{hyperref}

\newcommand{\pmdiags}{\mathfrak{D}_n^{\textnormal{\textsf{pm}}}}
\newcommand{\trdiags}{\mathfrak{D}_n^{\textnormal{\textsf{tr}}}}

\DeclareMathSymbol{\mlq}{\mathord}{operators}{``}
\DeclareMathSymbol{\mrq}{\mathord}{operators}{`'}

\newcommand{\downmapsto}{
  \rotatebox[origin=c]{-90}{$\scriptstyle\mapsto$}\mkern2mu
}

\newcommand{\fixwidth}[3]{\makebox[\widthof{#2}][#3]{#1}}

\newcommand{\hls}[1] {
    \tikz[baseline,
      outer sep=+15pt, inner sep = 2pt%
    ]
   \node[decorate,rectangle,minimum height=14pt,fill=black!16!white,anchor=text]{#1};
}

\tikzstyle{point} = [
  circle,
  fill = black,
  inner sep = 0pt,
  minimum width = 3.5pt
]
\tikzstyle{node} = [
  circle,
  draw = black,
  fill = white,
  inner sep = 0pt,
  minimum width = 3.5pt
]
\tikzstyle{inactivenode} = [
  circle,
  draw = black,
  fill = black!16!white,
  inner sep = 0pt,
  minimum width = 3.5pt
]
\tikzstyle{trap} = [
  densely dotted,
  black
]
\tikzstyle{trapedge} = [
  black
]
\newcommand{\trapinit}[3] {
  \coordinate (bottom) at (0,#2);
  \coordinate (roof) at (0,#1);
  \draw[name path=floor,opacity=0] (bottom) -- ++(#3,0);
  \draw[name path=ceil,opacity=0] (roof) -- ++(#3,0);
}
\newcommand{\trapbord}[3] {
  \draw[name path=upper,opacity=0] (#1) -- (#1 |- roof);
  \draw[name path=lower,opacity=0] (#1) -- (#1 |- bottom);
  \draw[trap,name intersections={of=#2 and upper}] (intersection-1) -- (#1);
  \draw[trap,name intersections={of=#3 and lower}] (intersection-1) -- (#1);
}
\newcommand{\trapbordlow}[2] {
  \draw[name path=lower,opacity=0] (#1) -- (#1 |- bottom);
  \draw[trap,name intersections={of=#2 and lower}] (intersection-1) -- (#1);
}
\newcommand{\trapbordhigh}[2] {
  \draw[name path=upper,opacity=0] (#1) -- (#1 |- roof);
  \draw[trap,name intersections={of=#2 and upper}] (intersection-1) -- (#1);
}

\newcommand{\trapedge}[3] {
  \draw[trapedge,name path global=#3] (#1) -- (#2);
}
\newcommand{\trapedgelow}[3] {
  \draw[trapedge,name path global=#3] (#1) .. controls ++(1,-3) and ++(-1,-3) .. (#2);
}
\newcommand{\trapedgehigh}[3] {
  \draw[trapedge,name path global=#3] (#1) .. controls ++(1,3) and ++(-1,3) .. (#2);
}
\newcommand{\trapedgelows}[3] {
  \draw[trapedge,name path global=#3] (#1) .. controls ++(1,-1.5) and ++(-1,-1.5) .. (#2);
}
\newcommand{\trapedgehighs}[3] {
  \draw[trapedge,name path global=#3] (#1) .. controls ++(1,1.5) and ++(-1,1.5) .. (#2);
}
\newcommand{\trapedgelowss}[3] {
  \draw[trapedge,name path global=#3] (#1) .. controls ++(0.75,-0.75) and ++(-0.75,-0.75) .. (#2);
}
\newcommand{\trapedgehighss}[3] {
  \draw[trapedge,name path global=#3] (#1) .. controls ++(0.75,0.75) and ++(-0.75,0.75) .. (#2);
}
\newcommand{\trapedgelowsss}[3] {
  \draw[trapedge,name path global=#3] (#1) .. controls ++(0.5,-0.5) and ++(-0.5,-0.5) .. (#2);
}

\begin{document}


\maketitle


\begin{abstract}

The $d$-dimensional Catalan numbers form a well-known sequence of numbers which count balanced bracket expressions over an alphabet of size $d$.
In this paper, we introduce and study what we call $d$-dimensional \emph{prime Catalan numbers}, a sequence of numbers which count only a very specific subset of indecomposable balanced bracket expressions.

We further introduce the notion of a \emph{trapezoidal diagram} of a crossing-free geometric graph, such as a triangulation or a crossing-free perfect matching.
In essence, such a diagram is obtained by augmenting the geometric graph in question with its trapezoidal decomposition, and then forgetting about the precise coordinates of individual vertices while preserving the vertical visibility relations between vertices and segments.
We note that trapezoidal diagrams of triangulations are closely related to abstract upward triangulations.

We study the numbers of such diagrams in the cases of (i) perfect matchings and (ii) triangulations.
We give bijective proofs which establish relations with 3-dimensional (prime) Catalan numbers.
This allows us to determine the corresponding exponential growth rates exactly as (i) $5.196^n$ and (ii) $23.459^n$ (bases are rounded to 3 decimal places).

Finally, we give lower bounds for the maximum number of embeddings of a trapezoidal diagram on any given point set.

\end{abstract}
\section{Introduction}

\paragraph{Prime Catalan Numbers.}
A \emph{balanced bracket expression (of dimension d)} is a finite string $c$ over an alphabet $\{b_1,\dots,b_d\}$ of $d$ \emph{brackets} such that all brackets occur in equal numbers in $c$, and such that every prefix of $c$ contains at least as many occurrences of $b_i$ as of $b_{i+1}$, for $1\leq i \leq d-1$.
The \emph{size} $m = |c|$ of $c$ is defined as the number of occurrences of $b_1$.
As an example, we enumerate all balanced bracket expressions of dimension $d=3$ (with brackets $b_1 = \mlq\langle\mrq$, $b_2 = \mlq|\mrq$, and $b_3 = \mlq\rangle\mrq$) and of size $m=2$.
\begin{align*}
  \langle|\langle\rangle|\rangle &&
  \langle\langle||\rangle\rangle &&
  \langle|\langle|\rangle\rangle &&
  \langle\langle|\rangle|\rangle &&
  \langle|\rangle\langle|\rangle
\end{align*}

We call a balanced bracket expression of dimension $d$ \emph{prime} if it does not contain any non-empty, contiguous and proper substrings that are themselves balanced bracket expressions (of dimension d).
Note that in the above enumeration, only the first two expressions are prime, whereas the other three all contain ``$\langle|\rangle$'' as a proper substring.

The $m$-th $d$-dimensional Catalan number $C^{(d)}_m$ can be defined as the number of balanced bracket expressions of dimension $d$ and of size $m$ \cite{GP13,S04}.
The most prominent instantiation of this family of sequences is of course given by the customary (2-dimensional) Catalan numbers $C_m = C^{(2)}_m$.
These numbers are ubiquitous in enumerative combinatorics, as illustrated by a famous exercise in Stanley's book with no less than 66 different combinatorial interpretations \cite[Exercise 6.19]{S01}.
In this paper, however, the focus will be on the 3-dimensional case \cite{L11}.
For easy reference, below we enumerate the first ten entries in the sequence corresponding to $C^{(3)}_m$, starting with $m=0$.
\begin{equation*}
  1, 1, 5, 42, 462, 6006, 87516, 1385670, 23371634, 414315330,\dots
\end{equation*}

Explicit product formulae for the numbers $C^{(d)}_m$ are known and can be obtained by employing the famous hook-length formula for standard Young tableaux of shape $(m^d)$ (see \cite{NPS97} for a precise statement and an insightful proof of the hook-length formula).
For example, for dimension $d=3$ we have
\begin{equation}\label{eq:catalan_three_exact}
  C^{(3)}_m = \frac{2(3m)!}{(m+2)!(m+1)!m!} \text{.}
\end{equation}

By employing Stirling's approximation for factorials, we further get the asymptotic estimate
\begin{equation}\label{eq:catalan_three_asympt}
  C^{(3)}_m \sim \frac{\sqrt{3}}{\pi} m^{-4} 27^m \text{ \; (as $m \rightarrow \infty$).}
\end{equation}

In a similar vein, we denote by $P^{(d)}_m$ the number of balanced bracket expressions of dimension $d$ and of size $m$ which are prime.
We call $P^{(d)}_m$ the $m$-th $d$-dimensional \emph{prime Catalan number}. 
In spite of the natural definition, we are not aware of any previous work that studies these numbers or acknowledges their existence.
Below we enumerate the first ten entries in the sequence corresponding to $P^{(3)}_m$, starting again with $m=0$.
\begin{equation*}
1, 1, 2, 12, 107, 1178, 14805, 203885, 3002973, 46573347, \dots
\end{equation*}

The above sequence stands in stark contrast to the customary 2-dimensional case.
Indeed, by reusing the brackets $b_1 = \mlq\langle\mrq$ and $b_2 = \mlq\rangle\mrq$, we easily see that the empty string $\varepsilon$ and ``$\langle\rangle$'' are the only prime balanced bracket expressions of dimension $d=2$.
Hence, $P^{(2)}_m = 0$ for $m \geq 2$.

In Section \ref{sec:primecatalan} we study various aspects of the numbers $P^{(d)}_m$ and we show how to compute them.
The theorem below gives the rate of exponential growth of the prime Catalan numbers, and it is the only result that is relevant for the earlier sections.
In the following, let $C^{(d)}(x) := \sum_{m=0}^\infty C^{(d)}_mx^m$ denote the ordinary generating function of the $d$-dimensional Catalan numbers.

\begin{theorem}
  \label{thm:primeasym}
  For any dimension $d \geq 3$, the prime Catalan numbers satisfy
  \begin{equation*}
    \lim_{m \rightarrow \infty} \sqrt[m]{P^{(d)}_m} = \gamma_d\text{, } \; \text{ where } \gamma_d := \left(\frac{d}{C^{(d)}(1/d^d)}\right)^d\text{.} 
  \end{equation*}
\end{theorem}

Decimal approximations of the numbers $\gamma_d$ can be computed automatically using any modern computer algebra system.
For $\gamma_3$ we even obtain a closed formula.
\begin{align}\label{eq:primegamma}
  \gamma_3 = \frac{27}{(\frac{729\sqrt{3}}{40\pi} - 9)^3} \approx 23.45948 && \gamma_4 \approx 251.78874 && \gamma_5 \approx 3119.93434
\end{align}

\paragraph{Trapezoidal Diagrams.}

Let $S$ be a set of $n$ points in the plane such that no three points are on a common line.
Additionally, we assume throughout that all points have distinct $x$-coordinates, which means that they can be ordered as $s_1,\dots,s_n$ from left to right.
We then say that a point $s_i$ is \emph{to the left} of another point $s_j$ if $i < j$ holds.

A \emph{crossing-free geometric graph} or \emph{plane graph} \emph{(on $S$)} is a graph $G$ with vertex set $S$ such that any two edges, which are drawn as straight-line segments between the corresponding endpoints, do not intersect except possibly at a common endpoint.
In this paper we restrict our attention to two special cases of crossing-free geometric graphs, namely \emph{perfect matchings} (1-regular crossing-free geometric graphs) and \emph{triangulations} (edge-maximal crossing-free geometric graphs). 
Even though we tend to omit the adjective \emph{crossing-free}, all geometric graphs considered in this paper have no crossings.

The \emph{trapezoidal decomposition} of a geometric graph $G$ is a well-known and useful notion (see \cite{Seidel10} for a classic application) which is obtained by drawing a vertical \emph{extension} upwards and downwards outgoing from each point in $S$ until a segment of $G$ is hit; if there is no obstruction, then the extension is drawn as an infinite ray.
If the extension going upwards (downwards) from a point $s$ hits the segment corresponding to an edge $e$, then we say that $e$ \emph{sees $s$ below (above)} in $G$.
Clearly, every point can be seen by at most two edges, once below and once above.

We now define the \emph{trapezoidal diagram} (or, just \emph{diagram}) of $G$, where $G$ is either a perfect matching or a triangulation.
Informally speaking, the trapezoidal diagram is equivalent to the trapezoidal decomposition except that we discard the coordinates of the vertices.

\begin{definition}\label{def:diagpm}
  Let $n$ be even, and let $M$ be a perfect matching on $S$.
  Then, the \emph{trapezoidal diagram} of $M$, denoted by $\mathcal{D}_M$, is defined as follows.
  \begin{itemize}
    \item
      $\mathcal{D}_M$ is an abstract graph with vertex set  $[n]$\footnote{We make use of the convenient notation $[n] = \{1,2,\dots,n\}$.} so that there is an edge $\{i,j\}$ in $\mathcal{D}_M$ if and only if there is an edge $\{s_i,s_j\}$ in $M$.
    \item
      Every edge in $\mathcal{D}_M$ has two distinguished sequences with the indices $i_1,\dots,i_k$ of the points $s_{i_1},\dots,s_{i_k}$ (sorted from left to right) that the corresponding edge in $M$ can see below and above, respectively.
    \item
      There are two additional sequences with the indices of the points (sorted from left to right) that no edge in $M$ can see below and above, respectively.
  \end{itemize}
  If there exists an isomorphism $[n] \rightarrow [n]$ between $\mathcal{D}_{M_1}$ and $\mathcal{D}_{M_2}$ that preserves the structure imposed by the above bullet points, then we identify $\mathcal{D}_{M_1}$ and $\mathcal{D}_{M_2}$, and we say that $M_1$ and $M_2$ have \emph{the same trapezoidal diagram}.
\end{definition}

We will typically not appeal to the above definition directly.
Instead, we will argue on the basis of a drawing of a trapezoidal diagram.
Incidentally, a \emph{drawing} of $\mathcal{D}_M$ is a plane (i.e., without crossings) drawing of the underlying graph and its trapezoidal decomposition, where we allow edges of the graph to be drawn as arbitrary $x$-monotone Jordan curves.
All orientations of edges must however remain the same (i.e., a left endpoint remains a left endpoint in the drawing) and all vertical visibility relations must remain identical (i.e., the order from left to right in which an edge sees points below does not change, and so on).

We refer to Figure \ref{fig:diagpm} for interesting examples.
Observe that the geometric graph $M$ combined with its trapezoidal decomposition is an instance of a drawing of $\mathcal{D}_M$.
Further note that two distinct perfect matchings $M_1$ and $M_2$ on the same point set $S$ may have the same trapezoidal diagram.

We remark at once that trapezoidal diagrams of perfect matchings are related to the well studied notion of \emph{(directed) bar visibility graphs} (see for example \cite{TT86} for a definition and further references).
However, that class of graphs imposes much coarser equivalence classes on the set of collections of non-intersecting segments in the plane.

\begin{figure}[t]
  \begin{center}
      \begin{tikzpicture}[scale=0.6,xscale=1.15]
    \begin{scope}[xshift=-40]
        \coordinate (c1) at (0.0,1);
        \coordinate (c2) at (2.2,0);
        \coordinate (c3) at (1.7,1.3);
        \coordinate (c4) at (3.5,1.7);
        \coordinate (c5) at (1.1,3);
        \coordinate (c6) at (2.8,2.5);
        \trapinit{3.35}{-0.35}{4}
        \trapedge{c1}{c2}{e1}
        \trapedge{c3}{c4}{e2}
        \trapedge{c5}{c6}{e3}
        \trapedge{c3}{c4}{e2}
          \trapbord{c1}{ceil}{floor};
          \trapbord{c2}{e2}{floor};
          \trapbord{c3}{e3}{e1};
          \trapbord{c4}{ceil}{floor};
          \trapbord{c5}{ceil}{e1};
          \trapbord{c6}{ceil}{e2};
        \node[point] at (c1) {};
        \node[point] at (c2) {};
        \node[point] at (c3) {};
        \node[point] at (c4) {};
        \node[point] at (c5) {};
        \node[point] at (c6) {};
     \end{scope}
      
     \begin{scope}[xshift=-40,yshift=-135]
        \coordinate (c1) at (0,0.5);
        \coordinate (c2) at (2.8,0.5);
        \coordinate (c3) at (1.4,1.5);
        \coordinate (c4) at (3.5,1.5);
        \coordinate (c5) at (0.7,2.5);
        \coordinate (c6) at (2.1,2.5);
        \trapinit{3.35}{-0.35}{4}
        \trapedge{c1}{c2}{e1}
        \trapedge{c3}{c4}{e2}
        \trapedge{c5}{c6}{e3}
          \trapbord{c1}{ceil}{floor};
          \trapbord{c2}{e2}{floor};
          \trapbord{c3}{e3}{e1};
          \trapbord{c4}{ceil}{floor};
          \trapbord{c5}{ceil}{e1};
          \trapbord{c6}{ceil}{e2};
        \node[node] at (c1) {};
        \node[node] at (c2) {};
        \node[node] at (c3) {};
        \node[node] at (c4) {};
        \node[node] at (c5) {};
        \node[node] at (c6) {};
      \end{scope}
      
    \begin{scope}[xshift=100]
        \coordinate (c1) at (0.0,1);
        \coordinate (c2) at (2.2,0);
        \coordinate (c3) at (1.7,1.3);
        \coordinate (c4) at (3.5,1.7);
        \coordinate (c5) at (1.1,3);
        \coordinate (c6) at (2.8,2.5);
        \trapinit{3.35}{-0.35}{4}
        \trapedge{c1}{c2}{e1}
        \trapedge{c3}{c4}{e2}
        \trapedge{c5}{c6}{e3}
          \trapbord{c1}{ceil}{floor};
          \trapbord{c2}{e2}{floor};
          \trapbord{c3}{e3}{e1};
          \trapbord{c4}{ceil}{floor};
          \trapbord{c5}{ceil}{e1};
          \trapbord{c6}{ceil}{e2};
        \node[node] at (c1) {};
        \node[node] at (c2) {};
        \node[node] at (c3) {};
        \node[node] at (c4) {};
        \node[node] at (c5) {};
        \node[node] at (c6) {};
     \end{scope}
     
     \begin{scope}[xshift=100,yshift=-135]
        \coordinate (c1) at (0.0,0.75);
        \coordinate (c2) at (2.8,0.75);
        \coordinate (c3) at (1.4,1.25);
        \coordinate (c4) at (3.5,1.25);
        \coordinate (c5) at (0.7,2.5);
        \coordinate (c6) at (2.1,2.5);
        \trapinit{3.35}{-0.35}{4}
        \trapedgelowss{c1}{c2}{e1}
        \trapedgehighss{c3}{c4}{e2}
        \trapedgelowsss{c5}{c6}{e3}
          \trapbord{c1}{ceil}{floor};
          \trapbord{c2}{e2}{floor};
          \trapbord{c3}{e3}{e1};
          \trapbord{c4}{ceil}{floor};
          \trapbord{c5}{ceil}{e1};
          \trapbord{c6}{ceil}{e2};
        \node[node] at (c1) {};
        \node[node] at (c2) {};
        \node[node] at (c3) {};
        \node[node] at (c4) {};
        \node[node] at (c5) {};
        \node[node] at (c6) {};
      \end{scope}
  
          \begin{scope}[xshift=300]
            \coordinate (c3)  at (1.1,0);
            \coordinate (c2)  at (2.9,0);
            \coordinate (c6)  at (1.4,0.75);
            \coordinate (c1)  at (2.6,0.75);
            \coordinate (c5)  at (0,1.5);
            \coordinate (c9)  at (0.8,1.5);
            \coordinate (c4)  at (3.2,1.5); 
            \coordinate (c8)  at (4,1.5);
            \coordinate (c12)  at (1.4,2.25);
            \coordinate (c7) at (2.6,2.25);
            \coordinate (c11) at (1.1,3);
            \coordinate (c10) at (2.9,3);
            \trapinit{3.35}{-0.35}{4}
            \trapedge{c1}{c2}{e1}
            \trapedge{c3}{c4}{e2}
            \trapedge{c5}{c6}{e3}
            \trapedge{c7}{c8}{e4}
            \trapedge{c9}{c10}{e5}
            \trapedge{c11}{c12}{e6}
            \trapbord{c1}{e2}{floor};
            \trapbord{c2}{e2}{floor};
            \trapbord{c3}{e3}{floor};
            \trapbord{c4}{e4}{floor};
            \trapbord{c5}{ceil}{floor};
            \trapbord{c6}{e5}{e2};
            \trapbord{c7}{e5}{e2};
            \trapbord{c8}{ceil}{floor};
            \trapbord{c9}{ceil}{e3};
            \trapbord{c10}{ceil}{e4};
            \trapbord{c11}{ceil}{e5};
            \trapbord{c12}{ceil}{e5};
            \node[point] at (c1) {};
            \node[point] at (c2) {};
            \node[point] at (c3) {};
            \node[point] at (c4) {};
            \node[point] at (c5) {};
            \node[point] at (c6) {};
            \node[point] at (c7) {};
            \node[point] at (c8) {};
            \node[point] at (c9) {};
            \node[point] at (c10) {};
            \node[point] at (c11) {};
            \node[point] at (c12) {};
          \end{scope}
          
          \begin{scope}[xshift=300,yshift=-135]
            \coordinate (c1)  at (1.1,0);
            \coordinate (c4)  at (2.9,0);
            \coordinate (c2)  at (1.4,0.75);
            \coordinate (c7)  at (2.6,0.75);
            \coordinate (c5)  at (0,1.5);
            \coordinate (c3)  at (0.8,1.5);
            \coordinate (c10)  at (3.2,1.5); 
            \coordinate (c8)  at (4,1.5);
            \coordinate (c6)  at (1.4,2.25);
            \coordinate (c11) at (2.6,2.25);
            \coordinate (c9) at (1.1,3);
            \coordinate (c12) at (2.9,3);
            \trapinit{3.35}{-0.35}{4}
            \trapedge{c1}{c2}{e1}
            \trapedge{c3}{c4}{e2}
            \trapedge{c5}{c6}{e3}
            \trapedge{c7}{c8}{e4}
            \trapedge{c9}{c10}{e5}
            \trapedge{c11}{c12}{e6}
            \trapbord{c1}{e2}{floor};
            \trapbord{c2}{e2}{floor};
            \trapbord{c3}{e3}{floor};
            \trapbord{c4}{e4}{floor};
            \trapbord{c5}{ceil}{floor};
            \trapbord{c6}{e5}{e2};
            \trapbord{c7}{e5}{e2};
            \trapbord{c8}{ceil}{floor};
            \trapbord{c9}{ceil}{e3};
            \trapbord{c10}{ceil}{e4};
            \trapbord{c11}{ceil}{e5};
            \trapbord{c12}{ceil}{e5};
            \node[point] at (c1) {};
            \node[point] at (c2) {};
            \node[point] at (c3) {};
            \node[point] at (c4) {};
            \node[point] at (c5) {};
            \node[point] at (c6) {};
            \node[point] at (c7) {};
            \node[point] at (c8) {};
            \node[point] at (c9) {};
            \node[point] at (c10) {};
            \node[point] at (c11) {};
            \node[point] at (c12) {};
          \end{scope}
  \end{tikzpicture}
  \end{center}
  \caption{On the left, a perfect matching with trapezoidal decomposition and three drawings of its trapezoidal diagram, which are distinguished visually from geometric graphs by leaving vertices blank. On the right, two perfect matchings on the same point set with the same trapezoidal diagram.}
  \label{fig:diagpm}
\end{figure}
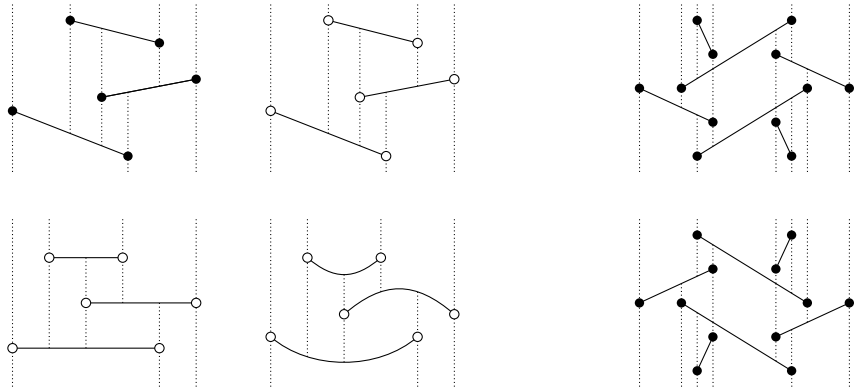

\begin{definition}\label{def:diagtr}
  Let $n \geq 2$ and assume (unless $n=2$) that $S$ has a triangular convex hull with the edge $\{s_1,s_n\}$ forming the lower envelope.
Then, the \emph{trapezoidal diagram} of a triangulation $T$ on $S$, denoted by $\mathcal{D}_T$, is obtained in the following way.
If $n \neq 2$, we first draw an additional edge between $s_1$ and $s_n$ as an $x$-monotone curve that goes above all other points and segments.
After that, $\mathcal{D}_T$ is defined analogously to Definition \ref{def:diagpm}.
\end{definition}

Inserting an additional edge between the left-most and right-most point might seem rather arbitrary.
However, apart from obtaining much nicer \emph{drawings}, this is done for a natural reason that will become clear later.

Here, a closely related concept from graph drawing is that of \emph{upward triangulations}.
These are abstract maximal planar graphs with directed edges such that there exists a plane embedding where all edges are drawn as $y$-monotone Jordan curves pointing upwards \cite{FGW12}.
After replacing \emph{upward} with \emph{rightward} (and $y$-monotone with $x$-monotone), and after looking at Figure \ref{fig:diagtr}, it becomes clear that every trapezoidal diagram of a geometric triangulation corresponds to a unique upward triangulation (by orienting edges from left to right) and that, depending on the presence or absence of symmetries, every upward triangulation corresponds to either one or two such diagrams.

\begin{figure}
  \begin{center}
       \begin{tikzpicture}[xscale=1.18]
      \begin{scope}[xshift=-10,yshift=-3,scale=0.5,yscale=1.5,scale=1.2,
          decoration={markings,mark=at position 0.6 with {\arrow{>}}}]
        \coordinate (c1) at (-1,0);
        \coordinate (c2) at (4,0);
        \coordinate (c3) at (0.5,0.66);
        \coordinate (c4) at (2.5,0.5);
        \coordinate (c5) at (1.5,2);
        \draw[postaction={decorate}] (c1) -- (c2);
        \draw[postaction={decorate}] (c1) -- (c3);
        \draw[postaction={decorate}] (c1) -- (c4);
        \draw[postaction={decorate}] (c3) -- (c4);
        \draw[postaction={decorate}] (c4) -- (c2);
        \draw[postaction={decorate}] (c1) -- (c5);
        \draw[postaction={decorate}] (c3) -- (c5);
        \draw[postaction={decorate}] (c5) -- (c4);
        \draw[postaction={decorate}] (c5) -- (c2);
        \node[point] at (c1) {};
        \node[point] at (c2) {};
        \node[point] at (c3) {};
        \node[point] at (c4) {};
        \node[point] at (c5) {};
      \end{scope}
      \begin{scope}[xshift=180,scale=0.4,yshift=60]
        \coordinate (c1) at (0,0);
        \coordinate (c2) at (5,0);
        \coordinate (c3) at (2.5,-1.25);
        \coordinate (c4) at (3.25,1.25);
        \coordinate (c5) at (1.75,0);
        \trapinit{2.5}{-2.5}{5}
        \trapedgelow{c1}{c2}{e1}
        \trapedgehigh{c1}{c2}{e10}
        \trapedge{c1}{c3}{e2}
        \trapedge{c3}{c2}{e3}
        \trapedge{c3}{c4}{e4}
        \trapedge{c1}{c5}{e5}
        \trapedge{c3}{c5}{e6}
        \trapedge{c4}{c5}{e7}
        \trapedge{c1}{c4}{e8}
        \trapedge{c4}{c2}{e9}
        \trapbord{c1}{ceil}{floor}
        \trapbord{c2}{ceil}{floor}
        \trapbord{c3}{e7}{e1}
        \trapbord{c4}{e10}{e3}
        \trapbord{c5}{e8}{e2}
        \node[node] at (c1) {};
        \node[node] at (c2) {};
        \node[node] at (c3) {};
        \node[node] at (c4) {};
        \node[node] at (c5) {};
      \end{scope}
      \begin{scope}[xshift=100,scale=0.4,yshift=60]
        \coordinate (c1) at (0,0);
        \coordinate (c2) at (5,0);
        \coordinate (c3) at (3.25,-1.25);
        \coordinate (c4) at (2.5,1.25);
        \coordinate (c5) at (1.75,0);
        \trapinit{2.5}{-2.5}{5}
        \trapedgelow{c1}{c2}{e1}
        \trapedgehigh{c1}{c2}{e10}
        \trapedge{c1}{c3}{e2}
        \trapedge{c3}{c2}{e3}
        \trapedge{c3}{c4}{e4}
        \trapedge{c1}{c5}{e5}
        \trapedge{c3}{c5}{e6}
        \trapedge{c4}{c5}{e7}
        \trapedge{c1}{c4}{e8}
        \trapedge{c4}{c2}{e9}
        \trapbord{c3}{e9}{e1}
        \trapbord{c4}{e10}{e6}
        \trapbord{c5}{e8}{e2}
        \trapbord{c1}{ceil}{floor}
        \trapbord{c2}{ceil}{floor}
        \node[node] at (c1) {};
        \node[node] at (c2) {};
        \node[node] at (c3) {};
        \node[node] at (c4) {};
        \node[node] at (c5) {};
      \end{scope}
    \end{tikzpicture}
  \end{center}
  \caption{An (abstract) upward (or rather, rightward) triangulation and two corresponding trapezoidal diagrams of (geometric) triangulations. The second diagram can be obtained by a vertical reflection.}
  \label{fig:diagtr}
\end{figure}
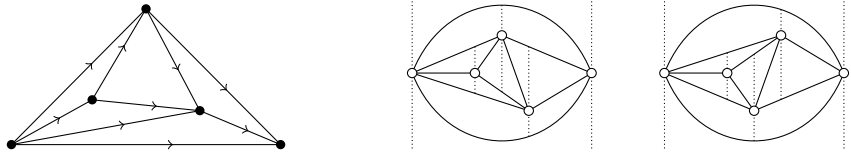

In Section \ref{sec:bijections} we present a generic method for encoding trapezoidal diagrams as strings over a finite alphabet.
The presented ideas can be applied to any family of geometric graphs.
However, we obtain a simple characterization of the set of code words only in the cases of perfect matchings and triangulations.

Let $\pmdiags = \{\mathcal{D}_M\}$ be the set of all trapezoidal diagrams of all perfect matchings on all point sets $S$ of size $n$.
See Table \ref{tab:diagpmenum} for an enumeration of this set for $n = 0,2,4$.

\begin{table}[b]
  \begin{center}
      \begin{tabular}{*2{>{\centering\arraybackslash}m{1.3cm}} >{\centering\arraybackslash}m{8.2cm}}
    \toprule
      $n=0$ & $n=2$ & $n=4$ \\
      \cmidrule(lr){1-1}\cmidrule(lr){2-2}\cmidrule(lr){3-3}
      \addlinespace[0.5em]
      \begin{tikzpicture}[scale=0.5]
        \node {$\varnothing$};
      \end{tikzpicture}
      &
      \begin{tikzpicture}[scale=0.5]
        \coordinate (c1) at (0.35,0.35);
        \coordinate (c2) at (1.75,0.35);
        \trapinit{1.05}{-0.35}{2.1}
        \trapedge{c1}{c2}{e1}
        \trapbord{c1}{ceil}{floor};
        \trapbord{c2}{ceil}{floor};
        \node[node] at (c1) {};
        \node[node] at (c2) {};
      \end{tikzpicture}
      &
      \begin{tikzpicture}[scale=0.5]
      \begin{scope}[xshift=-190]
        \coordinate (c1) at (0,0);
        \coordinate (c2) at (1.4,0);
        \coordinate (c3) at (0.7,0.7);
        \coordinate (c4) at (2.1,0.7);
        \trapinit{1.05}{-0.35}{2.1}
        \trapedge{c1}{c2}{e1}
        \trapedge{c3}{c4}{e2}
        \trapbord{c1}{ceil}{floor};
        \trapbord{c2}{e2}{floor};
        \trapbord{c3}{ceil}{e1};
        \trapbord{c4}{ceil}{floor};
        \node[node] at (c1) {};
        \node[node] at (c2) {};
        \node[node] at (c3) {};
        \node[node] at (c4) {};
      \end{scope}
      \begin{scope}[xshift=-95]
        \coordinate (c1) at (0.7,0);
        \coordinate (c2) at (2.1,0);
        \coordinate (c3) at (0,0.7);
        \coordinate (c4) at (1.4,0.7);
        \trapinit{1.05}{-0.35}{2.1}
        \trapedge{c1}{c2}{e1}
        \trapedge{c3}{c4}{e2}
        \trapbord{c1}{e2}{floor};
        \trapbord{c2}{ceil}{floor};
        \trapbord{c3}{ceil}{floor};
        \trapbord{c4}{ceil}{e1};
        \node[node] at (c1) {};
        \node[node] at (c2) {};
        \node[node] at (c3) {};
        \node[node] at (c4) {};
      \end{scope}
      \begin{scope}[xshift=0]
        \coordinate (c1) at (0,0);
        \coordinate (c2) at (2.1,0);
        \coordinate (c3) at (0.7,0.7);
        \coordinate (c4) at (1.4,0.7);
        \trapinit{1.05}{-0.35}{2.1}
        \trapedge{c1}{c2}{e1}
        \trapedge{c3}{c4}{e2}
        \trapbord{c1}{ceil}{floor};
        \trapbord{c2}{ceil}{floor};
        \trapbord{c3}{ceil}{e1};
        \trapbord{c4}{ceil}{e1};
        \node[node] at (c1) {};
        \node[node] at (c2) {};
        \node[node] at (c3) {};
        \node[node] at (c4) {};
      \end{scope}
      \begin{scope}[xshift=95]
        \coordinate (c1) at (0.7,0);
        \coordinate (c2) at (1.4,0);
        \coordinate (c3) at (0,0.7);
        \coordinate (c4) at (2.1,0.7);
        \trapinit{1.05}{-0.35}{2.1}
        \trapedge{c1}{c2}{e1}
        \trapedge{c3}{c4}{e2}
        \trapbord{c1}{e2}{floor};
        \trapbord{c2}{e2}{floor};
        \trapbord{c3}{ceil}{floor};
        \trapbord{c4}{ceil}{floor};
        \node[node] at (c1) {};
        \node[node] at (c2) {};
        \node[node] at (c3) {};
        \node[node] at (c4) {};
      \end{scope}
      \begin{scope}[xshift=190]
        \coordinate (c1) at (0,0.35);
        \coordinate (c2) at (0.7,0.35);
        \coordinate (c3) at (1.4,0.35);
        \coordinate (c4) at (2.1,0.35);
        \trapinit{1.05}{-0.35}{2.1}
        \trapedge{c1}{c2}{e1}
        \trapedge{c3}{c4}{e2}
        \trapbord{c1}{ceil}{floor};
        \trapbord{c2}{ceil}{floor};
        \trapbord{c3}{ceil}{floor};
        \trapbord{c4}{ceil}{floor};
        \node[node] at (c1) {};
        \node[node] at (c2) {};
        \node[node] at (c3) {};
        \node[node] at (c4) {};
      \end{scope}
      \end{tikzpicture} \\
      \addlinespace[0.25em]
      \bottomrule
    \end{tabular}
  \end{center}
  \caption{All non-isomorphic trapezoidal diagrams of perfect matchings, with $\varnothing$ denoting the empty diagram.}
  \label{tab:diagpmenum}
\end{table}

\begin{theorem}\label{thm:pmdiags}
  For any $n = 2m$ we have that the number of trapezoidal diagrams of perfect matchings on $n$ points is equal to the $m$-th 3-dimensional Catalan number, i.e.,
  \begin{equation*}
    |\pmdiags| = C^{(3)}_m = \frac{2(3m)!}{(m+2)!(m+1)!m!}
    \sim \beta n^{-4} \alpha^n \text{ \; (as $n \rightarrow \infty$),}
  \end{equation*}
  where
  \begin{equation*}
    \alpha = \sqrt{27} \approx 5.19615\text{, } \; \text{ and }  \beta = \frac{16\sqrt{3}}{\pi} \approx 8.82126 \text{.}
  \end{equation*}
\end{theorem}

Similarly, let $\trdiags = \{\mathcal{D}_T\}$ be the set of all trapezoidal diagrams of all triangulations on all sets $S$ of size $n$ as specified in Definition \ref{def:diagtr}.
Enumerations for the cases $n = 2,3,4,5$ can be seen in Table \ref{tab:diagtrenum}.

\begin{table}
  \begin{center}
    \input{tab/diagtrenum}
  \end{center}
  \caption{All non-isomorphic trapezoidal diagrams of triangulations. We omit drawing the vertical extensions through the left-most and right-most vertices.}
  \label{tab:diagtrenum}
\end{table}

\begin{theorem}\label{thm:trdiags}
  For any $n = m+2$ we have that the number of trapezoidal diagrams of triangulations on $n$ points is equal to the $m$-th 3-dimensional prime Catalan number, i.e.,\footnote{We make use of the $\Theta^\ast$-notation, which hides any unattributed subexponential factors.}
  \begin{equation*}
    |\trdiags| = P^{(3)}_m = \Theta^\ast(\gamma_3^n)\text{, } \; \text{ where } \gamma_3 = \frac{27}{(\frac{729\sqrt{3}}{40\pi} - 9)^3} \approx 23.45948\text{.}
  \end{equation*}
\end{theorem}

The following corollary should be compared with a classic result of Tutte \cite{Tutte62}, which implies that the number of abstract triangulations (i.e., maximal planar graphs) on $n$ vertices is $\Theta^\ast(\delta^n)$, for $\delta = 256/27 \approx 9.481$.
Moreover, quite curiously, in a side remark of \cite{AS13} the authors report an upper bound of $27^n$ on the number of upward triangulations.
As we shall see, the appearance of the exponential base $27$ is not at all coincidental.

\begin{corollary}
  Let $N_n$ be the number of abstract upward triangulations on $n$ vertices.
  Then, we have $|\trdiags|/2 \leq N_n \leq |\trdiags|$ and hence also $N_n = \Theta^\ast(\gamma_3^n)$.
\end{corollary}

\paragraph{Number of Embeddings.}

Originally, our interest in trapezoidal diagrams came from a desire for improved upper bounds on the maximum number of crossing-free geometric graphs on any set $S$ of $n$ points.
A classic result due to Ajtai et al.\ \cite{ACNS82} implies that, for any family of graphs, this maximum number is equal to $\Theta^\ast(\delta^n)$ for some absolute constant $\delta$.
Upper bounds for these numbers have been improved gradually over the past decades, culminating in $\delta \leq 10.05$ for perfect matchings \cite{SW06} and $\delta \leq 30$ for triangulations~\cite{SS11}.
However, there are no matching lower bounds, and the general consensus is that the known upper bounds are still far away from the truth.

In Section \ref{sec:embeddings} we initiate the study of the maximum number of embeddings that a given trapezoidal diagram can have on a fixed point set.
While we were able to find two simple exponential lower bounds, we did not succeed in proving strong enough upper bounds so as to obtain improved bounds for the aforementioned constants $\delta$.
\section{Encoding Trapezoidal Diagrams}\label{sec:bijections}

Let $G$ be a crossing-free geometric graph of one of the investigated types.
Fix a drawing of $\mathcal{D}_G$ and consider the set of all points in the plane which are neither a vertex, nor part of an edge, nor a vertical extension.
Then, a \emph{trapezoid} in $\mathcal{D}_G$ is defined as the closure of a maximal connected region in that set.
Typically, but not always, a trapezoid is bounded from above and below by (parts of) edges of $G$, and to the left and right by vertical extensions.

We further define a \emph{canonical order} over the trapezoids in $\mathcal{D}_G$ in the following recursive manner.
Given a prefix of the canonical order, we select as the next element a trapezoid that is either unbounded from below or that is bounded from below by an edge $e$ which is already well-supported, in the sense that all trapezoids having $e$ as their upper boundary occur in the given prefix of the canonical order.
If the above choice is not unique, then we settle with the left-most option.

By the following observation, which follows by induction over the length of the given prefix, the canonical order is seen to be both well-defined and independent of the fixed drawing of $\mathcal{D}_G$.

\begin{observation}\label{obs:staircase}
  Take any proper (both non-empty and incomplete) prefix of the canonical order of $\mathcal{D}_G$, build the union of all trapezoids in that prefix, and consider the boundary of that union.
  This boundary has a stair-case shape as depicted in Figure \ref{fig:canonical}.
  Specifically:
  \begin{itemize}
    \item Starting at positive infinity at the end of a vertical extension which is unbounded from above, the boundary alternates between verticals that go downwards and (parts of) not necessarily straight edges that go to the right, and it finally ends at negative infinity at the end of a vertical extension which is unbounded from below.
    \item Every vertical on the boundary contains exactly one vertex of $G$, either (a) at the bottom, (b) in its relative interior, or (c) at the top. When going along the boundary, we first encounter a (possibly empty) sequence of verticals of type (a), then at most one vertical of type (b), and then a (possibly empty) sequence of verticals of type (c).
  \end{itemize}
  Further note that the subsequent trapezoid in canonical order must be bounded to the left by the last vertical that is not of type (c).
\end{observation}

\begin{figure}
  \begin{center}
        \begin{tikzpicture}[scale=0.48,xscale=0.96]
      \begin{scope}[xshift=-30]
        \coordinate (c1) at (1,0);
        \coordinate (c2) at (3.2,0);
        \coordinate (c3) at (4.4,0);
        \coordinate (c4) at (8,0);
        \coordinate (c5) at (2,1);
        \coordinate (c6) at (5,1);
        \coordinate (c7) at (6,1);
        \coordinate (c8) at (7,1);
        \coordinate (c9) at (0,2);
        \coordinate (c10) at (2.6,2);
        \coordinate (c11) at (3.8,2);
        \coordinate (c12) at (9,2);
        \coordinate (c13) at (1.5,3);
        \coordinate (c14) at (6.5,3);
        \trapinit{3.75}{-0.75}{10}
        \trapedge{c1}{c2}{e1}
        \trapedge{c3}{c4}{e2}
        \trapedge{c5}{c6}{e3}
        \trapedge{c7}{c8}{e4}
        \trapedge{c9}{c10}{e5}
        \trapedge{c11}{c12}{e6}
        \trapedge{c13}{c14}{e7}
        \trapbord{c1}{e5}{floor}
        \trapbord{c2}{e3}{floor}
        \trapbord{c3}{e3}{floor}
        \trapbord{c4}{e6}{floor}
        \trapbord{c5}{e5}{e1}
        \trapbord{c6}{e6}{e2}
        \trapbord{c7}{e6}{e2}
        \trapbord{c8}{e6}{e2}
        \trapbord{c9}{ceil}{floor}
        \trapbord{c10}{e7}{e3}
        \trapbord{c11}{e7}{e3}
        \trapbord{c12}{ceil}{floor}
        \trapbord{c13}{ceil}{e5}
        \trapbord{c14}{ceil}{e6}
        \node[node] at (c1) {};
        \node[node] at (c2) {};
        \node[node] at (c3) {};
        \node[node] at (c4) {};
        \node[node] at (c5) {};
        \node[node] at (c6) {};
        \node[node] at (c7) {};
        \node[node] at (c8) {};
        \node[node] at (c9) {};
        \node[node] at (c10) {};
        \node[node] at (c11) {};
        \node[node] at (c12) {};
        \node[node] at (c13) {};
        \node[node] at (c14) {};
        \node at (-0.5,1.5) {\footnotesize 1};
        \node at (0.5,0.5) {\footnotesize 2};
        \node at (2.1,-0.5) {\footnotesize 3};
        \node at (1.5,1) {\footnotesize 4};
        \node at (2.6,0.5) {\footnotesize 5};
        \node at (3.8,0) {\footnotesize 6};
        \node at (6.2,-0.5) {\footnotesize 7};
        \node at (4.7,0.5) {\footnotesize 8};
        \node at (2.3,1.5) {\footnotesize 9};
        \node at (0.75,3) {\footnotesize 10};
        \node at (2.05,2.5) {\footnotesize 11};
        \node at (3.2,2) {\footnotesize 12};
        \node at (4.4,1.5) {\footnotesize 13};
        \node at (5.5,1) {\footnotesize 14};
        \node at (6.5,0.5) {\footnotesize 15};
        \node at (6.5,1.5) {\footnotesize 16};
        \node at (7.5,1) {\footnotesize 17};
        \node at (8.5,0.5) {\footnotesize 18};
        \node at (5.15,2.5) {\footnotesize 19};
        \node at (4,3.5) {\footnotesize 20};
        \node at (7.75,3) {\footnotesize 21};
        \node at (9.5,1.5) {\footnotesize 22};
      \end{scope}
      \begin{scope}[xshift=325]
        \coordinate (c1) at (1,0);
        \coordinate (c2) at (3.2,0);
        \coordinate (c3) at (4.4,0);
        \coordinate (c4) at (8,0);
        \coordinate (c5) at (2,1);
        \coordinate (c6) at (5,1);
        \coordinate (c7) at (6,1);
        \coordinate (c8) at (7,1);
        \coordinate (c9) at (0,2);
        \coordinate (c10) at (2.6,2);
        \coordinate (c11) at (3.8,2);
        \coordinate (c12) at (9,2);
        \coordinate (c13) at (1.5,3);
        \coordinate (c14) at (6.5,3);
        \fill[black!16!white]
          (0,3.75) -- (c9) -- (2,2) -- (c5) -- (3.2,1) --
          (c2) -- (3.2,-0.75) -- (-0.75,-0.75) -- (-0.75,3.75) -- cycle;
        \draw[trap] (0,3.75) -- (c9);
        \draw[trapedge] (2,2) -- (c9);
        \draw[trap] (2,2) -- (c5);
        \draw[trapedge] (3.2,1) -- (c5);
        \draw[trap] (3.2,1) -- (c2);
        \draw[trap] (3.2,-0.75) -- (c2);
        \node[node] at (c2) {};
        \node[node] at (c5) {};
        \node[node] at (c9) {};
      \end{scope}
      \begin{scope}[xshift=460]
        \coordinate (c1) at (1,0);
        \coordinate (c2) at (3.2,0);
        \coordinate (c3) at (4.4,0);
        \coordinate (c4) at (8,0);
        \coordinate (c5) at (2,1);
        \coordinate (c6) at (5,1);
        \coordinate (c7) at (6,1);
        \coordinate (c8) at (7,1);
        \coordinate (c9) at (0,2);
        \coordinate (c10) at (2.6,2);
        \coordinate (c11) at (3.8,2);
        \coordinate (c12) at (9,2);
        \coordinate (c13) at (1.5,3);
        \coordinate (c14) at (6.5,3);
        \fill[black!16!white]
          (1.5,3.75) -- (c13) -- (3.8,3) -- (c11) --
          (6,2) -- (c7) -- (c8) -- (7,0) -- (c4) --
          (8,-0.75) -- (-0.75,-0.75) -- (-0.75,3.75) -- cycle;
        \draw[trap] (1.5,3.75) -- (c13);
        \draw[trapedge] (3.8,3) -- (c13);
        \draw[trap] (3.8,3) -- (c11);
        \draw[trapedge] (6,2) -- (c11);
        \draw[trap] (6,2) -- (c7);
        \draw[trapedge] (c7) -- (c8);
        \draw[trap] (7,0) -- (c8);
        \draw[trapedge] (7,0) -- (c4);
        \draw[trap] (8,-0.75) -- (c4);
        
        \node[node] at (c13) {};
        \node[node] at (c11) {};
        \node[node] at (c7) {};
        \node[node] at (c8) {};
        \node[node] at (c4) {};
      \end{scope}
    \end{tikzpicture}
  \end{center}
  \caption{The canonical order of a trapezoidal diagram, and the boundaries corresponding to the prefixes $1,\dots,5$ and $1,\dots,15$.}
  \label{fig:canonical}
\end{figure}
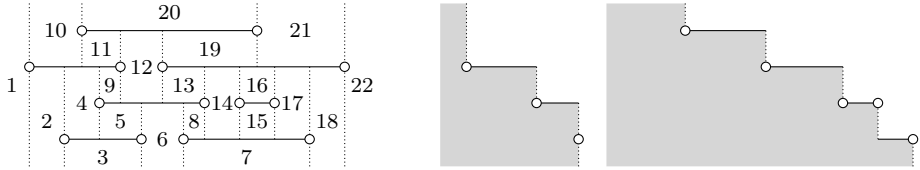

We are ready to prove the following lemma.
Combining it with equations (\ref{eq:catalan_three_exact}) and (\ref{eq:catalan_three_asympt}) yields Theorem \ref{thm:pmdiags}.

\begin{lemma}\label{lem:pmdiags}
  For any $n = 2m$ there is a bijection between $\pmdiags$ and the set of balanced bracket expressions of dimension 3 and of size $m$.
\end{lemma}
\begin{proof}
  For $m = 0$ the claim is trivial.
  So let $m \geq 1$, and let us define mappings in both directions.
  Observing that these mappings are inverses of each other concludes the proof.
  
  \paragraph{From trapezoids to brackets.}
  Let $\mathcal{D} = \mathcal{D}_M$ be the trapezoidal diagram of an arbitrary perfect matching $M$ on a set $S$ of $n$ points.
  We show how to construct the corresponding balanced bracket expression $c$ of size $m$.
  
  We first enumerate the trapezoids in $\mathcal{D}$ in canonical order, where we omit the last trapezoid on the far right.
  We obtain a sequence of exactly $3m$ trapezoids, each one of which is bounded to the right.
  Indeed, observe that to each edge $e = \{i,j\}$ in $\mathcal{D}$, where $i < j$ are the respective left and right endpoints, we can attribute the following three trapezoids.
  \begin{enumerate}[(i)]
    \item 
      The trapezoid whose right boundary is the vertical extension through $i$.
    \item 
      The trapezoid whose right boundary is the vertical extension below $j$.
    \item 
      The trapezoid whose right boundary is the vertical extension above $j$.
  \end{enumerate}
  
  In order to obtain $c$, we now apply the substitution rules $\text{(i)} \mapsto \mlq\langle\mrq$, $\text{(ii)} \mapsto \mlq|\mrq$, $\text{(iii)} \mapsto\mlq\rangle\mrq$, based on the three types of trapezoids specified above.
  The resulting string must be a balanced bracket expression of size $m$ because, as is clear from Observation \ref{obs:staircase}, for each edge $e$ the three attributed trapezoids occur in the relative order (i), (ii), (iii).
  
  \paragraph{From brackets to trapezoids.}
  Let $c$ be an arbitrary balanced bracket expression of dimension $3$ and of size $m$.
  We show how to construct the corresponding trapezoidal diagram $\mathcal{D} = \mathcal{D}_M$ of a perfect matching $M$ on $n$ points.
  
  We iterate over $c$ and construct a drawing of $\mathcal{D}$ by drawing one trapezoid per letter in $c$.
  For each bracket we select a different type of trapezoid.
  More precisely, as follows, we discriminate between the possible locations of the vertex $i$ that lies on the right boundary of the new trapezoid.
 
\newcommand{\eps}{0.15}
\newcommand{\trappmi}{
\begin{tikzpicture}[scale=0.45,baseline=(O.base)]
  \clip (0-\eps,0-\eps) rectangle (1+\eps,1+\eps);
  \node (O) at (0,+2mm) {};
  \fill[black!16!white] (0,0) -- (1,0) -- (1,1) -- (0,1) -- cycle;
  \draw[trapedge] (0,0) -- (1,0);
  \draw[trapedge] (0,1) -- (1,1);
  \draw[trap] (0,0) -- (0,1);
  \draw[trap] (1,0) -- (1,1);
  \node[node] at (1,0.5) {};
\end{tikzpicture}
}
\newcommand{\trappmii}{
\begin{tikzpicture}[scale=0.45,baseline=(O.base)]
  \clip (0-\eps,0-\eps) rectangle (1+\eps,1+\eps);
  \node (O) at (0,+2mm) {};
  \fill[black!16!white] (0,0) -- (1,0) -- (1,1) -- (0,1) -- cycle;
  \draw[trapedge] (0,0) -- (1,0);
  \draw[trapedge] (0,1) -- (1,1);
  \draw[trap] (0,0) -- (0,1);
  \draw[trap] (1,0) -- (1,1);
  \node[node] at (1,1) {};
\end{tikzpicture}
}
\newcommand{\trappmiii}{
\begin{tikzpicture}[scale=0.45,baseline=(O.base)]
  \clip (0-\eps,0-\eps) rectangle (1+\eps,1+\eps);
  \node (O) at (0,+2mm) {};
  \fill[black!16!white] (0,0) -- (1,0) -- (1,1) -- (0,1) -- cycle;
  \draw[trapedge] (0,0) -- (1,0);
  \draw[trapedge] (0,1) -- (1,1);
  \draw[trap] (0,0) -- (0,1);
  \draw[trap] (1,0) -- (1,1);
  \node[node] at (1,0) {};
\end{tikzpicture}
}
  \begin{enumerate}[(i)]
    \item 
      For \fixwidth{$\mlq\langle\mrq$}{$\mlq\langle\mrq$}{l} we select\trappmi, i.e., $i$ is in the interior of the right boundary.
    \item 
      For \fixwidth{$\mlq|\mrq$}{$\mlq\langle\mrq$}{l} we select\trappmii, i.e., $i$ is at the top of the right boundary.
    \item 
      For \fixwidth{$\mlq\rangle\mrq$}{$\mlq\langle\mrq$}{l} we select\trappmiii, i.e., $i$ is at the bottom of the right boundary.
  \end{enumerate}
  
  In the illustrations above we have omitted to draw the vertices on the respective left boundaries.
  Also, the trapezoids of types (i) and (ii), but not (iii), might in fact be unbounded from below.
  Similarly, the trapezoids of type (i) and (iii), but not (ii), might be unbounded from above.
  
  The positioning of individual trapezoids is done as illustrated in Table~\ref{tab:diagpmrec}, where the labels $l$ and $r$ indicate whether a boundary vertex is a left or right endpoint.
  Mutations which involve unbounded trapezoids can be handled analogously. 
  Also note that, after each step, the boundary of the union of all drawn trapezoids has a stair-case shape as in Observation \ref{obs:staircase}, and the order in which we add trapezoids corresponds to the canonical order.
 
\begin{table}
  \begin{center}
     \begin{tabular}{*4{>{\centering\arraybackslash}m{2.6cm}}}
    \toprule
      Before & \multicolumn{3}{c}{After} \\
      \addlinespace[0.25em]
      & $\mlq\langle\mrq$ & $\mlq|\mrq$ & $\mlq\rangle\mrq$ \\
    \cmidrule(lr){1-1}\cmidrule(lr){2-4}
    \addlinespace[0.25em]
      \begin{tikzpicture}[scale=0.5,xscale=1.35]
        \fill[black!16!white]
          (0,3) -- (1,3) -- (1,2.25) -- (2,2.25) --
          (2,0.75) -- (3,0.75) -- (3,0) -- (0,0) -- cycle;
        \draw[trap] (1,3) -- (1,2.25);
        \draw[trapedge] (1,2.25) -- (2,2.25);
        \draw[trap] (2,2.25) -- (2,0.75);
        \draw[trapedge] (2,0.75) -- (3,0.75);
        \draw[trap] (3,0.75) -- (3,0);
        \node[node,label=left:$l$] at (1,2.25) {};
        \node[inactivenode,label=left:$l$] at (2,0.75) {};
        \node[node,label=right:$r$] at (3,0.75) {};
      \end{tikzpicture}
      &
      \begin{tikzpicture}[scale=0.5,xscale=1.35]
        \fill[black!16!white]
          (0,3) -- (1,3) -- (1,2.25) -- (2.5,2.25) --
          (2.5,0.75) -- (3,0.75) -- (3,0) -- (0,0) -- cycle;
        \draw[trap] (1,3) -- (1,2.25);
        \draw[trapedge] (1,2.25) -- (2.5,2.25);
        \draw[trap] (2,2.25) -- (2,0.75);
        \draw[trap] (2.5,2.25) -- (2.5,0.75);
        \draw[trapedge] (2,0.75) -- (3,0.75);
        \draw[trap] (3,0.75) -- (3,0);
        \node[node,label=left:$l$] at (1,2.25) {};
        \node[inactivenode] at (2,0.75) {};
        \node[node,label=right:$r$] at (3,0.75) {};
        \node[node,label=right:$l$] at (2.5,1.5) {};
      \end{tikzpicture}
      &
      \begin{tikzpicture}[scale=0.5,xscale=1.35]
        \fill[black!16!white]
          (0,3) -- (1,3) -- (1,2.25) -- (2.5,2.25) --
          (2.5,0.75) -- (3,0.75) -- (3,0) -- (0,0) -- cycle;
        \draw[trap] (1,3) -- (1,2.25);
        \draw[trapedge] (1,2.25) -- (2.5,2.25);
        \draw[trap] (2,2.25) -- (2,0.75);
        \draw[trap] (2.5,2.25) -- (2.5,0.75);
        \draw[trapedge] (2,0.75) -- (3,0.75);
        \draw[trap] (3,0.75) -- (3,0);
        \node[inactivenode,label=left:$l$] at (1,2.25) {};
        \node[inactivenode] at (2,0.75) {};
        \node[node,label=right:$r$] at (3,0.75) {};
        \node[node,label=right:$r$] at (2.5,2.25) {};
      \end{tikzpicture}
      &
      \begin{tikzpicture}[scale=0.5,xscale=1.35]
        \fill[black!16!white]
          (0,3) -- (1,3) -- (1,2.25) -- (3,2.25) --
          (3,0) -- (0,0) -- cycle;
        \draw[trap] (1,3) -- (1,2.25);
        \draw[trapedge] (1,2.25) -- (3,2.25);
        \draw[trap] (3,2.25) -- (3,0.75);
        \draw[trap] (2,2.25) -- (2,0.75);
        \draw[trapedge] (2,0.75) -- (3,0.75);
        \draw[trap] (3,0.75) -- (3,0);
        \node[node,label=left:$l$] at (1,2.25) {};
        \node[inactivenode] at (2,0.75) {};
        \node[inactivenode,label=right:$r$] at (3,0.75) {};
      \end{tikzpicture}
      \\
      \addlinespace[0.25em]
      \begin{tikzpicture}[scale=0.5,xscale=1.35]
        \fill[black!16!white]
          (0,3) -- (1,3) -- (1,2.25) -- (2,2.25) --
          (2,0.75) -- (3,0.75) -- (3,0) -- (0,0) -- cycle;
        \draw[trap] (1,3) -- (1,2.25);
        \draw[trapedge] (1,2.25) -- (2,2.25);
        \draw[trap] (2,2.25) -- (2,0.75);
        \draw[trapedge] (2,0.75) -- (3,0.75);
        \draw[trap] (3,0.75) -- (3,0);
        \node[node,label=left:$l$] at (1,2.25) {};
        \node[inactivenode,label=left:$r$] at (2,1.5) {};
        \node[node,label=right:$r$] at (3,0.75) {};
      \end{tikzpicture}
      &
      \begin{tikzpicture}[scale=0.5,xscale=1.35]
        \fill[black!16!white]
          (0,3) -- (1,3) -- (1,2.25) -- (2.5,2.25) --
          (2.5,0.75) -- (3,0.75) -- (3,0) -- (0,0) -- cycle;
        \draw[trap] (1,3) -- (1,2.25);
        \draw[trapedge] (1,2.25) -- (2.5,2.25);
        \draw[trap] (2,2.25) -- (2,0.75);
        \draw[trap] (2.5,2.25) -- (2.5,0.75);
        \draw[trapedge] (2,0.75) -- (3,0.75);
        \draw[trap] (3,0.75) -- (3,0);
        \node[node,label=left:$l$] at (1,2.25) {};
        \node[inactivenode] at (2,1.5) {};
        \node[node,label=right:$r$] at (3,0.75) {};
        \node[node,label=right:$l$] at (2.5,1.5) {};
      \end{tikzpicture}
      &
      \begin{tikzpicture}[scale=0.5,xscale=1.35]
        \fill[black!16!white]
          (0,3) -- (1,3) -- (1,2.25) -- (2.5,2.25) --
          (2.5,0.75) -- (3,0.75) -- (3,0) -- (0,0) -- cycle;
        \draw[trap] (1,3) -- (1,2.25);
        \draw[trapedge] (1,2.25) -- (2.5,2.25);
        \draw[trap] (2,2.25) -- (2,0.75);
        \draw[trap] (2.5,2.25) -- (2.5,0.75);
        \draw[trapedge] (2,0.75) -- (3,0.75);
        \draw[trap] (3,0.75) -- (3,0);
        \node[inactivenode,label=left:$l$] at (1,2.25) {};
        \node[inactivenode] at (2,1.5) {};
        \node[node,label=right:$r$] at (3,0.75) {};
        \node[node,label=right:$r$] at (2.5,2.25) {};
      \end{tikzpicture}
      &
      \begin{tikzpicture}[scale=0.5,xscale=1.35]
        \fill[black!16!white]
          (0,3) -- (1,3) -- (1,2.25) -- (3,2.25) --
          (3,0) -- (0,0) -- cycle;
        \draw[trap] (1,3) -- (1,2.25);
        \draw[trapedge] (1,2.25) -- (3,2.25);
        \draw[trap] (3,2.25) -- (3,0.75);
        \draw[trap] (2,2.25) -- (2,0.75);
        \draw[trapedge] (2,0.75) -- (3,0.75);
        \draw[trap] (3,0.75) -- (3,0);
        \node[node,label=left:$l$] at (1,2.25) {};
        \node[inactivenode] at (2,1.5) {};
        \node[inactivenode,label=right:$r$] at (3,0.75) {};
      \end{tikzpicture}
      \\
      \addlinespace[0.25em]
      \begin{tikzpicture}[scale=0.5,xscale=1.35]
        \fill[black!16!white]
          (0,3) -- (1,3) -- (1,2.25) -- (2,2.25) --
          (2,0.75) -- (3,0.75) -- (3,0) -- (0,0) -- cycle;
        \draw[trap] (1,3) -- (1,2.25);
        \draw[trapedge] (1,2.25) -- (2,2.25);
        \draw[trap] (2,2.25) -- (2,0.75);
        \draw[trapedge] (2,0.75) -- (3,0.75);
        \draw[trap] (3,0.75) -- (3,0);
        \node[node,label=left:$l$] at (1,2.25) {};
        \node[node,label=left:$l$] at (2,1.5) {};
        \node[node,label=right:$r$] at (3,0.75) {};
      \end{tikzpicture}
      &
      \begin{tikzpicture}[scale=0.5,xscale=1.35]
        \fill[black!16!white]
          (0,3) -- (1,3) -- (1,2.25) -- (2,2.25) -- (2,1.5) -- (2.5,1.5) --
          (2.5,0.75) -- (3,0.75) -- (3,0) -- (0,0) -- cycle;
        \draw[trap] (1,3) -- (1,2.25);
        \draw[trapedge] (1,2.25) -- (2,2.25);
        \draw[trapedge] (2,1.5) -- (2.5,1.5);
        \draw[trap] (2,2.25) -- (2,0.75);
        \draw[trap] (2.5,1.5) -- (2.5,0.75);
        \draw[trapedge] (2,0.75) -- (3,0.75);
        \draw[trap] (3,0.75) -- (3,0);
        \node[node,label=left:$l$] at (1,2.25) {};
        \node[node,label=left:$l$] at (2,1.5) {};
        \node[node,label=45:$l$] at (2.5,1.125) {};
        \node[node,label=right:$r$] at (3,0.75) {};
      \end{tikzpicture}
      &
      \begin{tikzpicture}[scale=0.5,xscale=1.35]
        \fill[black!16!white]
          (0,3) -- (1,3) -- (1,2.25) -- (2,2.25) -- (2,1.5) -- (2.5,1.5) --
          (2.5,0.75) -- (3,0.75) -- (3,0) -- (0,0) -- cycle;
        \draw[trap] (1,3) -- (1,2.25);
        \draw[trapedge] (1,2.25) -- (2,2.25);
        \draw[trapedge] (2,1.5) -- (2.5,1.5);
        \draw[trap] (2,2.25) -- (2,0.75);
        \draw[trap] (2.5,1.5) -- (2.5,0.75);
        \draw[trapedge] (2,0.75) -- (3,0.75);
        \draw[trap] (3,0.75) -- (3,0);
        \node[node,label=left:$l$] at (1,2.25) {};
        \node[inactivenode,label=left:$l$] at (2,1.5) {};
        \node[node,label=right:$r$] at (2.5,1.5) {};
        \node[node,label=right:$r$] at (3,0.75) {};
      \end{tikzpicture}
      &
      \begin{tikzpicture}[scale=0.5,xscale=1.35]
        \fill[black!16!white]
          (0,3) -- (1,3) -- (1,2.25) -- (2,2.25) -- (2,1.5) -- (3,1.5) --
          (3,0) -- (0,0) -- cycle;
        \draw[trap] (1,3) -- (1,2.25);
        \draw[trapedge] (1,2.25) -- (2,2.25);
        \draw[trap] (2,2.25) -- (2,0.75);
        \draw[trapedge] (2,1.5) -- (3,1.5);
        \draw[trap] (3,1.5) -- (3,0.75);
        \draw[trapedge] (2,0.75) -- (3,0.75);
        \draw[trap] (3,0.75) -- (3,0);
        \node[node,label=left:$l$] at (1,2.25) {};
        \node[node,label=left:$l$] at (2,1.5) {};
        \node[inactivenode,label=right:$r$] at (3,0.75) {};
      \end{tikzpicture}
      \\
    \addlinespace[0.25em]
    \bottomrule
  \end{tabular}
  \end{center}
  \caption{Constructing the trapezoidal diagram of a perfect matching. Vertices which are not active left or right endpoints are drawn in gray.}
  \label{tab:diagpmrec}
\end{table}

  We now have to show that, if $c$ is a balanced bracket expression, then each trapezoid can be placed in a coherent way.
  Assume thus that we have processed a certain prefix of $c$ already.
  Then, each left endpoint on the current boundary, except for those on a vertical of type (a) directly followed by a vertical of type (c) (as specified in Observation~\ref{obs:staircase}), is called an \emph{active left endpoint}.
  Similarly, each right endpoint on the current boundary, except for those on a vertical of type (b), is called an \emph{active right endpoint}.
  Let $m_1$, $m_2$ and $m_3$ be the respective numbers of occurrences of the brackets $\mlq\langle\mrq$, $\mlq|\mrq$ and $\mlq\rangle\mrq$ in the processed prefix of $c$.
  We claim that we maintain the following invariants.
  \begin{enumerate}
    \item[(I1)] 
      The number of active left endpoints on the boundary is equal to $m_1 - m_2$.
    \item[(I2)] 
      The number of active right endpoints on the boundary is equal to $m_2 - m_3$.
  \end{enumerate}
  
  These invariants are a consequence of the following observations:
  Adding a trapezoid of type (i) creates a new active left endpoint.
  Adding a trapezoid of type (ii) turns a formerly active left endpoint inactive, and it also creates a new active right endpoint.
  Adding a trapezoid of type (iii) turns a formerly active right endpoint inactive.
  Again, refer to Table \ref{tab:diagpmrec} for helpful illustrations.
 
  The above invariants guarantee that we never get stuck when constructing $\mathcal{D}$.
  Indeed, if the current bracket to be processed is $\mlq\langle\mrq$, then it is always possible to add a trapezoid of type (i).
  If the current bracket is $\mlq|\mrq$, then we can add a trapezoid of type (ii) only if there is an active left endpoint on the boundary, which is guaranteed by (I1) because the already processed prefix of $c$ satisfies $m_1 > m_2$ (because $c$ is a balanced bracket expression).
  If the current bracket is $\mlq\rangle\mrq$, then we can add a trapezoid of type (iii) only if there is an active right endpoint on the boundary, which is guaranteed by (I2) because the already processed prefix of $c$ satisfies $m_2 > m_3$.
   
  Also, by invariants (I1) and (I2), when the whole string $c$ has been processed we end up with a boundary that consists of a single vertical with one inactive right endpoint.
  The last trapezoid (i.e., the one that is unbounded to the right) can then be added in order to finish the construction of $\mathcal{D}$.
  %
\end{proof}

The proof of the next lemma is very similar to the preceding one.
Combining the lemma with Theorem \ref{thm:primeasym} and equation (\ref{eq:primegamma}) yields Theorem \ref{thm:trdiags}.

\begin{lemma}\label{lem:trdiags}
  For any $n = m + 2$ there is a bijection between $\trdiags$ and the set of prime balanced bracket expressions of dimension 3 and of size $m$.
\end{lemma}
\begin{proof}
  Assume again that $m \geq 1$.
  We proceed by defining mappings in both directions which are clearly inverses of each other.
  
  \paragraph{From trapezoids to brackets.}
  Let $\mathcal{D} = \mathcal{D}_T$ be the trapezoidal diagram of a triangulation $T$ on a set of $n$ points, as specified in Definition \ref{def:diagtr}.
  We show how to construct the corresponding balanced bracket expression $c$ of size $m$.
  
  We start by enumerating the trapezoids in $\mathcal{D}$ in canonical order, where we only consider trapezoids that are enclosed by the double edge $\{1,n\}$.
  In other words, we ignore all four unbounded trapezoids.
  The reader should not be confused by the fact that all enumerated trapezoids have only one vertical boundary and hence look more like triangles.
  Further note that we get a sequence of $4m$ trapezoids in this way.
  Indeed, to each of the $m$ inner vertices $i \in [n] \setminus \{1,n\}$ we can attribute the following four trapezoids.
  \begin{enumerate}[(i)]
    \item 
      The trapezoid whose right boundary is the vertical extension below $i$.
    \item 
      The trapezoid whose right boundary is the vertical extension above $i$.
    \item 
      The trapezoid whose left boundary is the vertical extension below $i$.
    \item 
      The trapezoid whose left boundary is the vertical extension above $i$.
  \end{enumerate}
  
  As a consequence of Observation \ref{obs:staircase}, the trapezoids of type (ii) and (iii) attributed to a common vertex $i$ always appear consecutively in the order (ii), (iii).
  Therefore, similar to what we did in the proof of Lemma \ref{lem:pmdiags}, we construct $c$ by applying the substitution rules $\text{(i)} \mapsto \mlq\langle\mrq$, $\text{(ii),(iii)} \mapsto \mlq|\mrq$, $\text{(iv)} \mapsto\mlq\rangle\mrq$.
  Note that in the case of the second rule we effectively replace two trapezoids with one single bracket.
  Also, by Observation \ref{obs:staircase}, the four trapezoids attributed to a common vertex $i$ occur in the relative order (i), (ii), (iii), (iv), implying that $c$ is indeed a balanced bracket expression of size $m$.
  In the next paragraph we will further see that $c$ is prime.
  
  \paragraph{From brackets to trapezoids.}
  Let $c$ be an arbitrary balanced bracket expression of size $m$.
  For the time being, we do \emph{not} make the assumption that $c$ is prime.
  We will try (and gracefully fail) to construct the corresponding trapezoidal diagram $\mathcal{D} = \mathcal{D}_T$ of a triangulation $T$ on $n$ points.
  
  We start by drawing the two obvious initial unbounded trapezoids.
  We then iterate over $c$ and draw one or two trapezoids per letter in $c$.
  Depending on the brackets we make the following selections.

\newcommand{\epstr}{0.15}
\newcommand{\traptri}{
\begin{tikzpicture}[scale=0.45,baseline=(O.base)]
  \node (O) at (0,+2mm) {};
  \fill[black!16!white] (1,0) -- (1,1) -- (0,0.5) -- cycle;
  \draw[trapedge] (0,0.5) -- (1,0);
  \draw[trapedge] (0,0.5) -- (1,1);
  \draw[trap] (1,0) -- (1,1);
  \node[node] at (1,1) {};
\end{tikzpicture}
}
\newcommand{\traptrii}{
\begin{tikzpicture}[scale=0.45,baseline=(O.base)]
  \node (O) at (0,+2mm) {};
  \fill[black!16!white] (-0.5,0.75) -- (1.5,0.25) -- (0.5,0) -- (0.5,1) -- cycle;
  \draw[trapedge] (-0.5,0.75) -- (1.5,0.25);
  \draw[trapedge] (-0.5,0.75) -- (0.5,1);
  \draw[trapedge] (1.5,0.25) -- (0.5,0);
  \draw[trap] (0.5,1) -- (0.5,0.5);
  \draw[trap] (0.5,0) -- (0.5,0.5);
  \node[node] at (0.5,0.5) {};
\end{tikzpicture}
}
\newcommand{\traptriv}{
\begin{tikzpicture}[scale=0.45,baseline=(O.base)]
  \node (O) at (0,+2mm) {};
  \fill[black!16!white] (0,1) -- (1,0.5) -- (0,0) -- cycle;
  \draw[trapedge] (0,1) -- (1,0.5);
  \draw[trapedge] (0,0) -- (1,0.5);
  \draw[trap] (0,1) -- (0,0);
  \node[node] at (0,0) {};
\end{tikzpicture}
}
  \begin{enumerate}
    \item[(i)] 
      For $\mlq\langle\mrq$ we select\traptri\!, i.e., vertical on the right with vertex at the top.
    \item[(ii,iii)] 
      For $\mlq|\mrq$ we select\traptrii\!, i.e., a combination of two trapezoids.
    \item[(iv)] 
      For $\mlq\rangle\mrq$ we select\traptriv\!, i.e., vertical on the left with vertex at the bottom.
  \end{enumerate}
  
  As for the positioning of individual trapezoids, we do it again in the obvious way by trying to maintain the invariant that, after each step, the boundary has a stair-case shape as in Observation \ref{obs:staircase}.
  In fact, if we regard the addition of the two trapezoids of type (ii,iii) as one single step, then the boundary will never contain any verticals of type (b) (as specified in Observation \ref{obs:staircase}).
  Helpful illustrations can be seen in Table \ref{tab:diagtrrec}.
  
  \begin{table}
    \begin{center}
        \begin{tabular}{*4{>{\centering\arraybackslash}m{2.6cm}}}
    \toprule
      Before & \multicolumn{3}{c}{After} \\
      \addlinespace[0.25em]
      & $\mlq\langle\mrq$ & $\mlq|\mrq$ & $\mlq\rangle\mrq$ \\
    \cmidrule(lr){1-1}\cmidrule(lr){2-4}
    \addlinespace[0.25em]
      \begin{tikzpicture}[scale=0.45,xscale=1.1]
        \fill[black!16!white]
          (0,4) -- (1,4) -- (1,3.25) -- (2,3.25) --
          (2,2) -- (3,2) -- (3,0.75) -- (4,0.75) --
          (4,0) -- (0,0) -- cycle;
      \draw[trapedge] (1,3.25) -- (2,3.25);
      \draw[trapedge] (2,2) -- (3,2);
      \draw[trapedge] (3,0.75) -- (4,0.75);
      \draw[trap] (1,4) -- (1,3.25);
      \draw[trap] (2,3.25) -- (2,2);
      \draw[trap] (3,2) -- (3,0.75);
      \draw[trap] (4,0) -- (4,0.75);
      \node[node,label=left:$l$] at (1,3.25) {};
      \node[inactivenode,label=left:$l$] at (2,2) {};
      \node[inactivenode,label=right:$r$] at (3,2) {};
      \node[node,label=right:$r$] at (4,0.75) {};
      \end{tikzpicture}
      &
      \begin{tikzpicture}[scale=0.45,xscale=1.1]
        \fill[black!16!white]
          (0,4) -- (1,4) -- (1,3.25) -- (2,3.25) --
          (2,2) -- (2.5,3) -- (2.5,2) -- (3,2) -- (3,0.75) -- (4,0.75) --
          (4,0) -- (0,0) -- cycle;
      \draw[trapedge] (1,3.25) -- (2,3.25);
      \draw[trapedge] (2,2) -- (3,2);
      \draw[trapedge] (2,2) -- (2.5,3);
      \draw[trapedge] (3,0.75) -- (4,0.75);
      \draw[trap] (1,4) -- (1,3.25);
      \draw[trap] (2,3.25) -- (2,2);
      \draw[trap] (2.5,2) -- (2.5,3);
      \draw[trap] (3,2) -- (3,0.75);
      \draw[trap] (4,0) -- (4,0.75);
      \node[node,label=left:$l$] at (1,3.25) {};
      \node[inactivenode,label=left:$l$] at (2,2) {};
      \node[node,label=right:$r$] at (3,2) {};
      \node[node,label=right:$r$] at (4,0.75) {};
      \node[inactivenode,label=right:$r$] at (2.5,3) {};
      \end{tikzpicture}
      &
      \begin{tikzpicture}[scale=0.45,xscale=1.1]
        \fill[black!16!white]
          (0,4) -- (1,4) -- (1,3.25) -- (2,3.25) --
          (2,2) -- (3,3) -- (3,2) -- (4,0.75) --
          (4,0) -- (0,0) -- cycle;
      \draw[trapedge] (1,3.25) -- (2,3.25);
      \draw[trapedge] (2,2) -- (3,2);
      \draw[trapedge] (2,2) -- (3,3);
      \draw[trapedge] (3,2) -- (4,0.75);
      \draw[trapedge] (3,0.75) -- (4,0.75);
      \draw[trap] (1,4) -- (1,3.25);
      \draw[trap] (2,3.25) -- (2,2);
      \draw[trap] (3,3) -- (3,2);
      \draw[trap] (3,2) -- (3,0.75);
      \draw[trap] (4,0) -- (4,0.75);
      \node[node,label=left:$l$] at (1,3.25) {};
      \node[node,label=left:$l$] at (2,2) {};
      \node[inactivenode,label=right:$l$] at (3,2) {};
      \node[inactivenode,label=right:$r$] at (4,0.75) {};
      \end{tikzpicture}
      &
      \begin{tikzpicture}[scale=0.45,xscale=1.1]
        \fill[black!16!white]
          (0,4) -- (1,4) -- (1,3.25) --
          (2,3.25) to[out=0,in=100,looseness=1]
          (3,2) -- (3,0.75) -- (4,0.75) --
          (4,0) -- (0,0) -- cycle;
      \draw[trapedge] (1,3.25) -- (2,3.25);
      \draw[trapedge] (2,2) -- (3,2);
      \draw[trapedge] (3,0.75) -- (4,0.75);
      \draw[trapedge] (2,3.25) to[out=0,in=100,looseness=1] (3,2);
      \draw[trap] (1,4) -- (1,3.25);
      \draw[trap] (2,3.25) -- (2,2);
      \draw[trap] (3,2) -- (3,0.75);
      \draw[trap] (4,0) -- (4,0.75);
      \node[inactivenode,label=left:$l$] at (1,3.25) {};
      \node[inactivenode] at (2,2) {};
      \node[inactivenode,label=right:$r$] at (3,2) {};
      \node[node,label=right:$r$] at (4,0.75) {};
      \end{tikzpicture}
      \\
    \addlinespace[0.25em]
    \bottomrule
  \end{tabular}
    \end{center}
    \caption{Constructing the trapezoidal diagram of a triangulation. Vertices which are not active left or right endpoints are drawn in gray.}
    \label{tab:diagtrrec}
  \end{table}

  Assume now that we have processed a certain prefix of $c$ already.   
  Every vertex on a vertical of type (a), except for the right-most one, is called an \emph{active left endpoint}.
  Similarly, every vertex on a vertical of type (c), except for the left-most one, is called an \emph{active right endpoint}.
  For $m_1$, $m_2$ and $m_3$ as in the proof of Lemma \ref{lem:pmdiags},
  we claim that we maintain the following invariants.
  \begin{enumerate}
    \item[(I1)] 
      The number of active right endpoints on the boundary is equal to $m_1 - m_2$.
    \item[(I2)] 
      The number of active left endpoints on the boundary is equal to $m_2 - m_3$.
  \end{enumerate}
  
  These invariants once more follow from three simple observations:
  Adding a trapezoid of type (i) turns a formerly inactive right endpoint active.
  Adding a pair of trapezoids of type (ii,iii) turns a formerly active right endpoint inactive, and it also turns a formerly inactive left endpoint active.
  Adding a trapezoid of type (iv) turns a formerly active left endpoint inactive.
  
  The above invariants guarantee that we never get stuck when constructing $\mathcal{D}$, even if $c$ is not prime.
  Indeed, if the current bracket to be processed is $\mlq\langle\mrq$, then we can always add a trapezoid of type (i).
  If the current bracket is $\mlq|\mrq$, then we can add a trapezoid of type (ii,iii) only if there is an active right endpoint, which is guaranteed by (I1) because the already processed prefix of $c$ satisfies $m_1 > m_2$.
  If the current bracket is $\mlq\rangle\mrq$, then we can add a trapezoid of type (iv) only if there is an active left endpoint, which is guaranteed by (I2) because the already processed prefix of $c$ satisfies $m_2 > m_3$.
  
  Furthermore, when $c$ has been processed completely, invariants (I1) and (I2) imply that the boundary consists of a single edge and two unbounded verticals (in other words, the staircase consists of a single step).
  Hence, we can just add the two final unbounded trapezoids in order to finish the construction of $\mathcal{D}$.
  
  Figure \ref{fig:doubleedges} shows that not every balanced bracket expression $c$ is mapped to a valid trapezoidal diagram.
  It can happen that double edges are created.
  Recall that one double edge between vertices $1$ and $n$ is required, but any other double edge or even a triple edge between vertices $1$ and $n$ is not in accordance with Definition \ref{def:diagtr}.
  All the same, we now see that the above reconstruction procedure creates a double edge whenever it finishes processing a substring of $c$ that is itself a balanced bracket expression.
  Since the described mapping clearly computes the inverse of the mapping from the preceding paragraph, this also implies that all balanced bracket expressions produced by that first mapping are in fact prime, as claimed earlier.
  
\begin{figure}
  \begin{center}
          \begin{tikzpicture}
        \begin{scope}[scale=0.25,xshift=0]
           \node at (3,7) {$\langle|\langle\rangle|\rangle$};
           \node at (3,4) {$\downmapsto$};
           \coordinate (c1) at (0,0);
           \coordinate (c2) at (6,0);
           \coordinate (c3) at (2,-1);
           \coordinate (c4) at (4,1);
           \trapinit{2.5}{-2.5}{6}
           \trapedgelow{c1}{c2}{e1}
           \trapedge{c1}{c3}{e2}
           \trapedge{c3}{c2}{e3}
           \trapedge{c3}{c4}{e4}
           \trapedge{c1}{c4}{e5}
           \trapedge{c4}{c2}{e6}
           \trapedgehigh{c1}{c2}{e7}
           \trapbord{c3}{e5}{e1}
           \trapbord{c4}{e7}{e3}
           \node[node] at (c1) {};
           \node[node] at (c2) {};
           \node[node] at (c3) {};
           \node[node] at (c4) {};
         \end{scope}
         \begin{scope}[scale=0.25,xshift=280]
           \node at (3,7) {$\langle\langle||\rangle\rangle$};
           \node at (3,4) {$\downmapsto$};
           \coordinate (c1) at (0,0);
           \coordinate (c2) at (6,0);
           \coordinate (c3) at (2,1);
           \coordinate (c4) at (4,-1);
           \trapinit{2.5}{-2.5}{6}
           \trapedgelow{c1}{c2}{e1}
           \trapedge{c1}{c3}{e2}
           \trapedge{c3}{c2}{e3}
           \trapedge{c3}{c4}{e4}
           \trapedge{c1}{c4}{e5}
           \trapedge{c4}{c2}{e6}
           \trapedgehigh{c1}{c2}{e7}
           \trapbord{c3}{e7}{e5}
           \trapbord{c4}{e3}{e1}
           \node[node] at (c1) {};
           \node[node] at (c2) {};
           \node[node] at (c3) {};
           \node[node] at (c4) {};
         \end{scope}
         \begin{scope}[scale=0.25,xshift=560]
           \node at (3,7) {$\langle| \, \hls{$\langle|\rangle$} \, \rangle$};
           \node at (3,4) {$\downmapsto$};
           \coordinate (c1) at (0,0);
           \coordinate (c2) at (6,0);
           \coordinate (c3) at (2,0);
           \coordinate (c4) at (4,0);
           \fill[black!16!white] (c3) .. controls ++(1,-1.5) and ++(-1,-1.5) .. (c2) .. controls ++(-1,1.5) and ++(1,1.5) .. (c3);
           \trapinit{2.5}{-2.5}{6}
           \trapedgelow{c1}{c2}{e1}
           \trapedge{c1}{c3}{e2}
           \trapedgelows{c3}{c2}{e3}
           \trapedge{c3}{c4}{e4}
           \trapedge{c4}{c2}{e5}
           \trapedgehighs{c3}{c2}{e6}
           \trapedgehigh{c1}{c2}{e7}
           \trapbord{c3}{e7}{e1}
           \trapbord{c4}{e6}{e3}
           \node[node] at (c1) {};
           \node[node] at (c2) {};
           \node[node] at (c3) {};
           \node[node] at (c4) {};
         \end{scope}
         \begin{scope}[scale=0.25,xshift=840]
           \node at (3,7) {$\langle \, \hls{$\langle | \rangle$} \, |\rangle$};
           \node at (3,4) {$\downmapsto$};
           \coordinate (c1) at (0,0);
           \coordinate (c2) at (6,0);
           \coordinate (c3) at (2,0);
           \coordinate (c4) at (4,0);
           \fill[black!16!white] (c1) .. controls ++(1,-1.5) and ++(-1,-1.5) .. (c4) .. controls ++(-1,1.5) and ++(1,1.5) .. (c1);
           \trapinit{2.5}{-2.5}{6}
           \trapedgelow{c1}{c2}{e1}
           \trapedgelows{c1}{c4}{e2}
           \trapedge{c1}{c3}{e3}
           \trapedge{c3}{c4}{e4}
           \trapedgehighs{c1}{c4}{e5}
           \trapedge{c4}{c2}{e6}
           \trapedgehigh{c1}{c2}{e7}
           \trapbord{c3}{e5}{e2}
           \trapbord{c4}{e7}{e1}
           \node[node] at (c1) {};
           \node[node] at (c2) {};
           \node[node] at (c3) {};
           \node[node] at (c4) {};
         \end{scope}
         \begin{scope}[scale=0.25,xshift=1120]
           \node at (3,7) {$\langle|\rangle \, \hls{$\langle|\rangle$}$};
           \node at (3,4) {$\downmapsto$};
           \coordinate (c1) at (0,0);
           \coordinate (c2) at (6,0);
           \coordinate (c3) at (3,-1.2);
           \coordinate (c4) at (3,1.2);
           \fill[black!16!white] (c1) -- (c2) .. controls ++(-1,3) and ++(1,3) .. (c1);
           \trapinit{2.5}{-2.5}{6}
           \trapedgelow{c1}{c2}{e1}
           \trapedge{c1}{c3}{e2}
           \trapedge{c3}{c2}{e3}
           \trapedge{c1}{c2}{e4}
           \trapedge{c1}{c4}{e5}
           \trapedge{c4}{c2}{e6}
           \trapedgehigh{c1}{c2}{e7}
           \trapbord{c3}{e4}{e1}
           \trapbord{c4}{e7}{e4}
           \node[node] at (c1) {};
           \node[node] at (c2) {};
           \node[node] at (c3) {};
           \node[node] at (c4) {};
         \end{scope}
      \end{tikzpicture}
  \end{center}
  \caption{Substrings which are balanced bracket expressions lead to unwanted double edges.}
  \label{fig:doubleedges}
\end{figure}
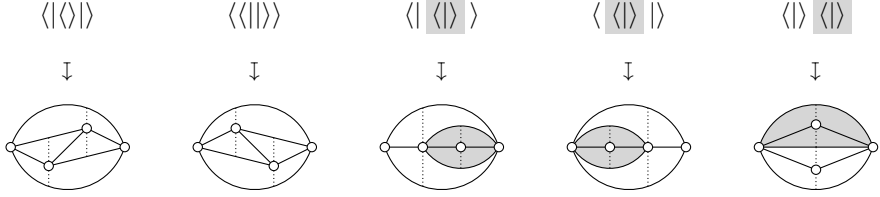
  
  Lastly, we face the problem of stretchability, namely that the produced drawing $\mathcal{D}$ might not correspond to the trapezoidal diagram $\mathcal{D}_T$ of an actual triangulation $T$ (with straight-line segments).
  However, it is known that any simple plane graph with edges drawn as non-crossing $x$-monotone curves can be stretched without changing edge orientations with respect to the $x$-axis \cite[Theorem 4]{Kelly87}.
  If $c$ is prime, it thus follows that also our drawing $\mathcal{D}$ is stretchable after removing the upper copy of the double edge $\{1,n\}$.
\end{proof}
\section{Embeddings of Trapezoidal Diagrams}\label{sec:embeddings}

Fix a trapezoidal diagram $\mathcal{D}$ with $n$ vertices and a set $S$ of $n$ points.
An \emph{embedding} of $\mathcal{D}$ on $S$ is a crossing-free geometric graph $G$ on $S$ with $\mathcal{D}_G = \mathcal{D}$.
Recall that for any family of crossing-free geometric graphs, the maximum number of such graphs on any set of $n$ points is equal to $\Theta^\ast(\delta^n)$ for some constant $\delta$.
If embeddings of $\mathcal{D}$ on every fixed $S$ had turned out to be unique, then Theorems \ref{thm:pmdiags} and \ref{thm:trdiags} would have implied the improved upper bounds $\delta \leq 5.196$ for perfect matchings and $\delta \leq 23.459$ for triangulations.
However, since embeddings are not unique in general, a natural follow-up question asks for the maximum number of embeddings.
While so far we did not succeed to obtain adequate upper bounds for these quantities, we can present two simple exponential lower bounds.

As already seen in Figure \ref{fig:diagpm}, there is a trapezoidal diagram of a perfect matching with $m=6$ edges that can be embedded in two different ways on a set of $n=12$ points.
By repeating that construction side by side as illustrated in Figure \ref{fig:pmmanyemb}, we get the following amplification.

\begin{figure}
  \begin{center}
      \begin{tikzpicture}[scale=0.6,xscale=1.1]
          \begin{scope}[xshift=0]
            \coordinate (a3)  at (1.1,0);
            \coordinate (a2)  at (2.9,0);
            \coordinate (a6)  at (1.4,0.75);
            \coordinate (a1)  at (2.6,0.75);
            \coordinate (a5)  at (0,1.5);
            \coordinate (a9)  at (0.8,1.5);
            \coordinate (a4)  at (3.2,1.5); 
            \coordinate (a8)  at (4,1.5);
            \coordinate (a12)  at (1.4,2.25);
            \coordinate (a7) at (2.6,2.25);
            \coordinate (a11) at (1.1,3);
            \coordinate (a10) at (2.9,3);
            \node[point] at (a1) {};
            \node[point] at (a2) {};
            \node[point] at (a3) {};
            \node[point] at (a4) {};
            \node[point] at (a5) {};
            \node[point] at (a6) {};
            \node[point] at (a7) {};
            \node[point] at (a9) {};
            \node[point] at (a10) {};
            \node[point] at (a11) {};
            \node[point] at (a12) {};
          \end{scope}
          \begin{scope}[xshift=100]
            \coordinate (b1)  at (1.1,0);
            \coordinate (b4)  at (2.9,0);
            \coordinate (b2)  at (1.4,0.75);
            \coordinate (b7)  at (2.6,0.75);
            \coordinate (b5)  at (0,1.5);
            \coordinate (b3)  at (0.8,1.5);
            \coordinate (b10)  at (3.2,1.5); 
            \coordinate (b8)  at (4,1.5);
            \coordinate (b6)  at (1.4,2.25);
            \coordinate (b11) at (2.6,2.25);
            \coordinate (b9) at (1.1,3);
            \coordinate (b12) at (2.9,3);
            \node[point] at (b1) {};
            \node[point] at (b2) {};
            \node[point] at (b3) {};
            \node[point] at (b4) {};
            \node[point] at (b6) {};
            \node[point] at (b7) {};
            \node[point] at (b9) {};
            \node[point] at (b10) {};
            \node[point] at (b11) {};
            \node[point] at (b12) {};
          \end{scope}
          \begin{scope}[xshift=200]
            \coordinate (c1)  at (1.1,0);
            \coordinate (c4)  at (2.9,0);
            \coordinate (c2)  at (1.4,0.75);
            \coordinate (c7)  at (2.6,0.75);
            \coordinate (c5)  at (0,1.5);
            \coordinate (c3)  at (0.8,1.5);
            \coordinate (c10)  at (3.2,1.5); 
            \coordinate (c8)  at (4,1.5);
            \coordinate (c6)  at (1.4,2.25);
            \coordinate (c11) at (2.6,2.25);
            \coordinate (c9) at (1.1,3);
            \coordinate (c12) at (2.9,3);
            \node[point] at (c1) {};
            \node[point] at (c2) {};
            \node[point] at (c3) {};
            \node[point] at (c4) {};
            \node[point] at (c6) {};
            \node[point] at (c7) {};
            \node[point] at (c9) {};
            \node[point] at (c10) {};
            \node[point] at (c11) {};
            \node[point] at (c12) {};
          \end{scope}
          \begin{scope}[xshift=300]
            \coordinate (d3)  at (1.1,0);
            \coordinate (d2)  at (2.9,0);
            \coordinate (d6)  at (1.4,0.75);
            \coordinate (d1)  at (2.6,0.75);
            \coordinate (d5)  at (0,1.5);
            \coordinate (d9)  at (0.8,1.5);
            \coordinate (d4)  at (3.2,1.5); 
            \coordinate (d8)  at (4,1.5);
            \coordinate (d12)  at (1.4,2.25);
            \coordinate (d7) at (2.6,2.25);
            \coordinate (d11) at (1.1,3);
            \coordinate (d10) at (2.9,3);
            \node[point] at (d1) {};
            \node[point] at (d2) {};
            \node[point] at (d3) {};
            \node[point] at (d4) {};
            \node[point] at (d6) {};
            \node[point] at (d7) {};
            \node[point] at (d9) {};
            \node[point] at (d10) {};
            \node[point] at (d11) {};
            \node[point] at (d12) {};
          \end{scope}
          \begin{scope}[xshift=400]
            \coordinate (e3)  at (1.1,0);
            \coordinate (e2)  at (2.9,0);
            \coordinate (e6)  at (1.4,0.75);
            \coordinate (e1)  at (2.6,0.75);
            \coordinate (e5)  at (0,1.5);
            \coordinate (e9)  at (0.8,1.5);
            \coordinate (e4)  at (3.2,1.5); 
            \coordinate (e8)  at (4,1.5);
            \coordinate (e12)  at (1.4,2.25);
            \coordinate (e7) at (2.6,2.25);
            \coordinate (e11) at (1.1,3);
            \coordinate (e10) at (2.9,3);
            \node[point] at (e1) {};
            \node[point] at (e2) {};
            \node[point] at (e3) {};
            \node[point] at (e4) {};
            \node[point] at (e6) {};
            \node[point] at (e7) {};
            \node[point] at (e8) {};
            \node[point] at (e9) {};
            \node[point] at (e10) {};
            \node[point] at (e11) {};
            \node[point] at (e12) {};
          \end{scope}
            \trapinit{3.35}{-0.35}{18.1}
            \trapedge{a1}{a2}{aa1}
            \trapedge{a3}{a4}{aa2}
            \trapedge{a5}{a6}{aa3}
            \trapedge{a7}{b6}{aa4}
            \trapedge{a9}{a10}{aa5}
            \trapedge{a11}{a12}{aa6}
            \trapbord{a1}{aa2}{floor};
            \trapbord{a2}{aa2}{floor};
            \trapbord{a3}{aa3}{floor};
            \trapbord{a4}{aa4}{floor};
            \trapbord{a5}{ceil}{floor};
            \trapbord{a6}{aa5}{aa2};
            \trapbord{a7}{aa5}{aa2};
            \trapbord{a9}{ceil}{aa3};
            \trapbord{a10}{ceil}{aa4};
            \trapbord{a11}{ceil}{aa5};
            \trapbord{a12}{ceil}{aa5};
            
            \trapedge{b1}{b2}{bb1}
            \trapedge{b3}{b4}{bb2}
            \trapedge{a7}{b6}{bb3}
            \trapedge{b7}{c6}{bb4}
            \trapedge{b9}{b10}{bb5}
            \trapedge{b11}{b12}{bb6}
            \trapbord{b1}{bb2}{floor};
            \trapbord{b2}{bb2}{floor};
            \trapbord{b3}{bb3}{floor};
            \trapbord{b4}{bb4}{floor};
            \trapbord{b6}{bb5}{bb2};
            \trapbord{b7}{bb5}{bb2};
            \trapbord{b9}{ceil}{bb3};
            \trapbord{b10}{ceil}{bb4};
            \trapbord{b11}{ceil}{bb5};
            \trapbord{b12}{ceil}{bb5};
            
            \trapedge{c1}{c2}{cc1}
            \trapedge{c3}{c4}{cc2}
            \trapedge{b7}{c6}{cc3}
            \trapedge{c7}{d6}{cc4}
            \trapedge{c9}{c10}{cc5}
            \trapedge{c11}{c12}{cc6}
            \trapbord{c1}{cc2}{floor};
            \trapbord{c2}{cc2}{floor};
            \trapbord{c3}{cc3}{floor};
            \trapbord{c4}{cc4}{floor};
            \trapbord{c6}{cc5}{cc2};
            \trapbord{c7}{cc5}{cc2};
            \trapbord{c9}{ceil}{cc3};
            \trapbord{c10}{ceil}{cc4};
            \trapbord{c11}{ceil}{cc5};
            \trapbord{c12}{ceil}{cc5};
            
            \trapedge{d1}{d2}{dd1}
            \trapedge{d3}{d4}{dd2}
            \trapedge{c7}{d6}{dd3}
            \trapedge{d7}{e6}{dd4}
            \trapedge{d9}{d10}{dd5}
            \trapedge{d11}{d12}{dd6}
            \trapbord{d1}{dd2}{floor};
            \trapbord{d2}{dd2}{floor};
            \trapbord{d3}{dd3}{floor};
            \trapbord{d4}{dd4}{floor};
            \trapbord{d6}{dd5}{dd2};
            \trapbord{d7}{dd5}{dd2};
            \trapbord{d9}{ceil}{dd3};
            \trapbord{d10}{ceil}{dd4};
            \trapbord{d11}{ceil}{dd5};
            \trapbord{d12}{ceil}{dd5};
            
            \trapedge{e1}{e2}{ee1}
            \trapedge{e3}{e4}{ee2}
            \trapedge{d7}{e6}{ee3}
            \trapedge{e7}{e8}{ee4}
            \trapedge{e9}{e10}{ee5}
            \trapedge{e11}{e12}{ee6}
            \trapbord{e1}{ee2}{floor};
            \trapbord{e2}{ee2}{floor};
            \trapbord{e3}{ee3}{floor};
            \trapbord{e4}{ee4}{floor};
            \trapbord{e6}{ee5}{ee2};
            \trapbord{e7}{ee5}{ee2};
            \trapbord{e8}{ceil}{floor};
            \trapbord{e9}{ceil}{ee3};
            \trapbord{e10}{ceil}{ee4};
            \trapbord{e11}{ceil}{ee5};
            \trapbord{e12}{ceil}{ee5};
  \end{tikzpicture}
  \end{center}
  \caption{A point set where a large number of distinct perfect matchings have the same trapezoidal diagram. Only one of $2^5 = 32$ such perfect matchings is shown.}
  \label{fig:pmmanyemb}
\end{figure}
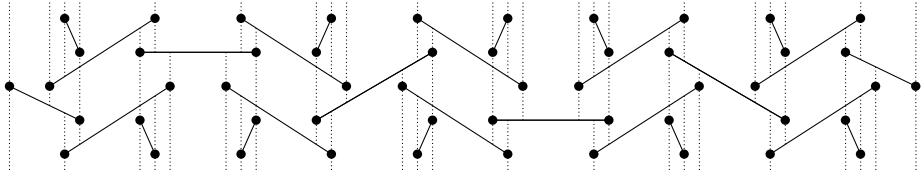

\begin{theorem}
  For any $k$ there exists a planar point set $S_k$ of size $n = 10k + 2$ and a diagram $\mathcal{D} \in \pmdiags$ with $2^k = \Omega(1.071^n)$ distinct embeddings on $S_k$.
\end{theorem}
\begin{proof}
  To construct $S_k$, put $k$ \emph{blocks}, each one consisting of a copy of the point set from Figure \ref{fig:diagpm}, side by side, but draw the respective left-most and right-most points only once, as exemplified in Figure \ref{fig:pmmanyemb} for $k=5$.
  In this way we get $10$ points per block and $2$ extra points, giving a total of $10k + 2$ points.
  The diagram $\mathcal{D}$ is chosen as a natural extension of the one seen in Figure \ref{fig:diagpm}.
  Observe now that for each block we can choose two distinct ways to embed the corresponding part of $\mathcal{D}$.
  Furthermore, these binary choices can be made independently, implying the desired number of $2^k$ embeddings.
\end{proof}

For triangulations we present an analogous construction.
It is based on the point set depicted in Figure \ref{fig:trmanyemb}, which is an adaptation of a point set taken from \cite{AS13} due to Günter Rote.
Originally, it was used to show that embeddings of upward triangulations on a given point set are not always unique.

\begin{theorem}\label{thm:manyisotrdiags}
  For any $k$ there exists a planar point set $S_k$ of size $n = 12k^2 + 4k + 3$ and a diagram $\mathcal{D} \in \trdiags$ with $2^{k^2} = \Omega(1.059^n)$ distinct embeddings on $S_k$.
\end{theorem}
\begin{proof}
  Define a \emph{block} as a copy of the point set depicted in Figure \ref{fig:trmanyemb}.
  Arrange $k^2$ such blocks in a honey comb grid, where extreme points of individual blocks may coincide with extreme points of neighboring blocks.
  Place three additional points such that $S_k$ has a triangular convex hull.
  It can be checked that this gives a total of $12k^{2} + 4k + 3$ points.
  The diagram $\mathcal{D}$ is chosen accordingly as a honey comb grid consisting of $k^2$ copies of the diagram depicted in Figure \ref{fig:trmanyemb} and some extra edges for connecting the hull vertices.
  The desired number of embeddings again follows after observing that we have independent binary choices for embedding individual blocks.
\end{proof}

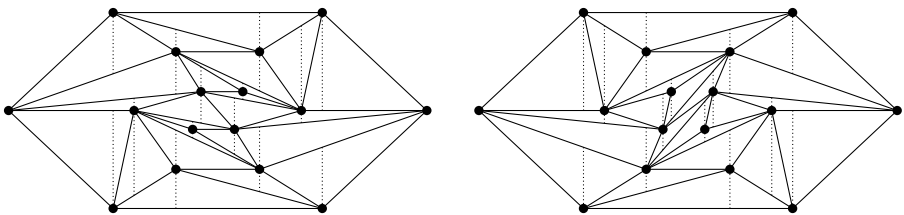
\begin{figure}[h]
  \begin{center}
      \begin{tikzpicture}[xscale=1.1]
    \newcommand{\pointsetout}[6]{
      \begin{scope}[yscale=0.6]
      \coordinate (#1) at (180:2.5);
      \coordinate (#2) at (120:2.5);
      \coordinate (#3) at (240:2.5);
      \coordinate (#4) at (60:2.5);
      \coordinate (#5) at (300:2.5);
      \coordinate (#6) at (0:2.5);
      \node[point] at (#1) {};
      \node[point] at (#2) {};
      \node[point] at (#3) {};
      \node[point] at (#4) {};
      \node[point] at (#5) {};
      \node[point] at (#6) {};
      \end{scope}
    }
    \newcommand{\pointsetmid}[6]{
      \begin{scope}[yscale=0.9]
      \coordinate (#1) at (180:1);
      \coordinate (#2) at (120:1);
      \coordinate (#3) at (240:1);
      \coordinate (#4) at (60:1);
      \coordinate (#5) at (300:1);
      \coordinate (#6) at (0:1);
      \node[point] at (#1) {};
      \node[point] at (#2) {};
      \node[point] at (#3) {};
      \node[point] at (#4) {};
      \node[point] at (#5) {};
      \node[point] at (#6) {};
      \end{scope}
    }
    \newcommand{\pointsetinn}[4]{
      \coordinate (#1) at (-0.3,-0.25);
      \coordinate (#2) at (-0.2,0.25);
      \coordinate (#3) at (0.2,-0.25);
      \coordinate (#4) at (0.3,0.25);
      \node[point] at (#1) {};
      \node[point] at (#2) {};
      \node[point] at (#3) {};
      \node[point] at (#4) {};
    }
    \newcommand{\edges}{
      \trapedge{c1}{c2}{e12}
      \trapedge{c1}{c5}{e15}
      \trapedge{c1}{c8}{e18}
      \trapedge{c1}{c4}{e14}
      \trapedge{c1}{c3}{e13}
      \trapedge{c2}{c14}{e214}
      \trapedge{c2}{c11}{e211}
      \trapedge{c2}{c5}{e25}
      \trapedge{c3}{c4}{e34}
      \trapedge{c3}{c6}{e36}
      \trapedge{c3}{c15}{e315}
      \trapedge{c4}{c8}{e48}
      \trapedge{c4}{c9}{e49}
      \trapedge{c4}{c7}{e47}
      \trapedge{c4}{c12}{e412}
      \trapedge{c4}{c6}{e46}
      \trapedge{c5}{c11}{e511}
      \trapedge{c5}{c13}{e513}
      \trapedge{c5}{c10}{e510}
      \trapedge{c5}{c8}{e58}
      \trapedge{c6}{c12}{e612}
      \trapedge{c6}{c15}{e615}
      \trapedge{c7}{c9}{e79}
      \trapedge{c7}{c12}{e712}
      \trapedge{c8}{c10}{e810}
      \trapedge{c8}{c13}{e813}
      \trapedge{c8}{c9}{e89}
      \trapedge{c9}{c13}{e913}
      \trapedge{c9}{c16}{e916}
      \trapedge{c9}{c12}{e912}
      \trapedge{c10}{c13}{e1013}
      \trapedge{c11}{c14}{e1114}
      \trapedge{c11}{c13}{e1113}
      \trapedge{c12}{c16}{e1216}
      \trapedge{c12}{c15}{e1215}
      \trapedge{c13}{c14}{e1314}
      \trapedge{c13}{c16}{e1316}
      \trapedge{c14}{c16}{e1416}
      \trapedge{c15}{c16}{e1516}
    }
    \newcommand{\traps}{
      \trapinit{1.3}{-1.3}{2.5}
      \trapbordlow{c2}{e15}
      \trapbordhigh{c3}{e14}
      \trapbord{c4}{e18}{e36}
      \trapbord{c5}{e211}{e18}
      \trapbord{c6}{e412}{e315}
      \trapbord{c7}{e49}{e412}
      \trapbord{c8}{e510}{e49}
      \trapbord{c9}{e813}{e712}
      \trapbord{c10}{e513}{e813}
      \trapbord{c11}{e214}{e513}
      \trapbord{c12}{e916}{e615}
      \trapbord{c13}{e1114}{e916}
      \trapbordlow{c14}{e1316}
      \trapbordhigh{c15}{e1216}
    }
    
    \begin{scope}
      \pointsetout{c1}{c2}{c3}{c14}{c15}{c16}
      \pointsetmid{c4}{c5}{c6}{c11}{c12}{c13}
      \pointsetinn{c7}{c8}{c9}{c10}
      \edges
      \traps
    \end{scope}
    \begin{scope}[xshift=160]
      \pointsetout{c1}{c2}{c3}{c14}{c15}{c16}
      \pointsetmid{c5}{c11}{c4}{c13}{c6}{c12}
      \pointsetinn{c8}{c10}{c7}{c9}
      \edges
      \traps
    \end{scope}
  \end{tikzpicture}
  \end{center}
  \caption{A point set and two triangulations which have the same trapezoidal diagram.}\label{fig:trmanyemb}
\end{figure}
\section{Prime Catalan Numbers}\label{sec:primecatalan}

In this final section we present all ingredients that are required to prove Theorem~\ref{thm:primeasym} from the introduction.
Furthermore, we show how to compute prime Catalan numbers efficiently.

In order to make notation less cumbersome when dealing with (prime) Catalan numbers of arbitrary dimension, we omit writing the superscripts $(d)$, but we always keep in mind the dependency on $d$.
That is, we write $C_m = C^{(d)}_m$ and $P_m = P^{(d)}_m$.
We further define the ordinary generating functions
\begin{align}
  C(x) = \sum_{m=0}^\infty C_m x^m\text{,} && P(x) = \sum_{m=0}^\infty P_m x^m\text{.}
\end{align}

We will be using a fundamental result of complex function theory called the Lagrange inversion formula.
In its classic form, it gives the Taylor expansion of the inverse of an analytic function at a point where the first derivative does not vanish.
In combinatorics, the following formulation is often most useful \cite[Theorem A.2]{FS09}.
\begin{theorem}[Lagrange Inversion]\label{thm:lagrange}
  Let $A(x) = \sum_{m=0}^\infty A_mx^m$ be a formal power series satisfying $A_0 \neq 0$. Define $Z(x) = \frac{x}{A(x)}$.
  Then, there exists a unique compositional inverse of Z(x), i.e., a unique formal power series $X(z) = \sum_{m=0}^\infty X_m z^m$ with $Z(X(z)) = z$.
  Moreover, the coefficients of $X(z)$ and $X(z)^k$ are given by
  \begin{align*}
    [z^m]X(z) = \frac{1}{m} [x^{m-1}] A(x)^m, && [z^m]X(z)^k = \frac{k}{m} [x^{m-k}] A(x)^m \text{.}
  \end{align*}
\end{theorem}

We will also be using a multiplicative variant of Fekete's lemma.
For a proof we refer the reader to \cite[Lemma 11.6]{LW01}.
\begin{theorem}[Fekete's Lemma]\label{thm:fekete}
  Let $A_0,A_1,A_2,\dots$ be a sequence of non-negative real numbers such that $A_{m+n} \geq A_m \cdot A_n$ holds for all $m,n$.
  Then,
  \begin{equation*}
    \lim_{m \rightarrow \infty} \sqrt[m]{A_m} = \limsup_{m \rightarrow \infty} \sqrt[m]{A_m}.
  \end{equation*}
  In particular, the limit exists unless it diverges to infinity.
\end{theorem}

We start by proving the following lemma, which establishes a formal relation between $C(x)$ and $P(x)$, and hence between the numbers $C_m$ and $P_m$.
\begin{lemma}\label{lem:relation}
  For any dimension $d$, the formal equality
    $C(x) = P(xC(x)^d)$ holds.
\end{lemma}
\begin{proof}
  Let $c$ be a balanced bracket expression of dimension $d$.
  Consider now all inclusion-maximal and contiguous substrings of $c$ which are themselves balanced bracket expressions (of dimension $d$) and which start someplace after the first letter of $c$.
  Call these substrings $c_1,c_2,\dots,c_k$ and note that some of them might be empty (see Figure \ref{fig:factorization} for an example for $d=3$).
  In fact, by definition, we have $|c_1| = |c_2| = \dots = |c_k| = 0$ if and only if $c$ is prime.
  
  \begin{figure}
    \begin{center}
      \begin{tikzpicture}
  \clip (-4.45,0.3) rectangle (4.5,-1.96);
  \begin{scope}[yshift=0]
    \node{
      $c = \langle\,\hls{$\langle|\rangle$}\,|\langle\,\hls{$\langle|\rangle$}\,\rangle|\,\hls{$\langle|\langle|\rangle\rangle$}\,\rangle\,\hls{$\langle|\rangle$}$
    };
  \end{scope}
  \begin{scope}[yshift=-39]
    \node{
      $
      \begin{aligned}
        p   &= \langle|\langle\rangle|\rangle &&&
        c_1 &= \hls{$\langle|\rangle$} &&&
        c_2 &= \varepsilon &&&
        c_3 &= \hls{$\langle|\rangle$}  \\
        |p| &= 2 &&&
        c_4 &= \varepsilon &&&
        c_5 &= \hls{$\langle|\langle|\rangle\rangle$} &&&
        c_6 &= \hls{$\langle|\rangle$}
      \end{aligned}
      $
    };
  \end{scope}
\end{tikzpicture}
    \end{center}
    \caption{A balanced bracket expression $c$, and the corresponding factorization consisting of $p$ and $c_1,\dots,c_6$.}
    \label{fig:factorization}
  \end{figure}
  
  Clearly, for $i \neq j$, $c_i$ and $c_j$ cannot be adjacent in $c$ since they are inclusion-maximal by assumption.
  Nor does $c_i$ contain $c_j$ or vice versa.
  Nor do they overlap because if that were the case, both their intersection and their union would be balanced bracket expressions, again contradicting inclusion-maximality.
  
  Therefore, after removing $c_1,\dots,c_k$ from $c$, we obtain a balanced bracket expression $p$ that is prime, whose size satisfies $d \cdot |p| = k$, and which yields back $c$ if $c_1,\dots,c_k$ are plugged back into the $k$ gaps in $p$ in the appropriate order (we ignore the ``gap'' before the first letter in $p$).
  Loosely speaking, the ordered collection consisting of $p$ and $c_1,\dots,c_k$ can be seen as a unique factorization of $c$.
  Further note that $|c| = |p| + |c_1| + \dots + |c_k|$.
  
  In the sums below, by letting the variables $c$ and $c_1,\dots,c_k$ run over all balanced bracket expressions of dimension $d$, and by letting $p$ run over all expressions that are prime, we now see that, indeed,
  \begin{align*}
    C(x) & = \sum_c x^{|c|}
         = \sum_p \sum_{\substack{c_1,\dots,c_k \\ k = d \cdot |p|}} x^{|p| + |c_1| + \dots + |c_k|} \\
         & = \sum_p x^{|p|}
             \sum_{\mathclap{\substack{c_1,\dots,c_k \\ k = d \cdot |p|}}} x^{|c_1| + \dots + |c_k|}
         = \sum_p x^{|p|} C(x)^{d \cdot |p|}
         = P(xC(x)^d) \text{.} \qedhere
  \end{align*}
\end{proof}

By combining Lemma \ref{lem:relation} with the Lagrange inversion formula, we obtain an efficient method for computing prime Catalan numbers of any dimension.
\begin{lemma}\label{lem:explicit}
  For any dimension $d$, we have
    $\displaystyle P_m = \frac{1}{1-dm} \cdot [x^m] \frac{1}{C(x)^{dm-1}}$.
\end{lemma}
\begin{proof}
  Define the formal power series
  \begin{align*}
    A(x) := \frac{1}{C(x)^d}\text{,} && Z(x) := \frac{x}{A(x)} = xC(x)^d \text{.}
  \end{align*}
  By Theorem \ref{thm:lagrange}, there exists $X(z)$ with $Z(X(z)) = z$.
  Hence, substituting $X(z)$ for $x$ in Lemma \ref{lem:relation} yields $C(X(z)) = P(z)$.
  Observe now that for $m=0$ the lemma holds because we have $P_0 = C_0 = 1$.
  For $m > 0$, by using the formula from Theorem \ref{thm:lagrange} in the fourth step,
  \begin{align*}
    P_m & = [z^m] P(z) = [z^m] C(X(z))
          = \sum_{k=0}^\infty C_k \cdot [z^m] X(z)^k \\
        & = \sum_{k=0}^\infty C_k \cdot \frac{k}{m} [x^{m-k}] A(x)^n
          = \frac{1}{m} \cdot [x^m] A(x)^m \cdot \underbrace{\sum_{k=0}^\infty k C_k x^k}_{= \, xC'(x)} \\
        & = \frac{1}{m} \cdot [x^{m-1}] \frac{C'(x)}{C(x)^{dm}}
          = \frac{1}{m} \cdot [x^{m-1}] \left(\frac{1}{1-dm} \frac{1}{C(x)^{dm-1}}\right)' \\
        & = \frac{1}{1-dm} \cdot [x^m] \frac{1}{C(x)^{dm-1}} \text{.} \qedhere
  \end{align*}
\end{proof}

In the introduction we observed that for $d=2$, the prime Catalan numbers do not give us a particularly exciting sequence.
For higher dimensions the situation is very different, as shown by the next lemma.

\begin{lemma}\label{lem:supermulti}
  For any dimension $d \geq 3$, the prime Catalan numbers are super-multiplicative, i.e., $P_{m+n} \geq P_m \cdot P_n$ for all $m$ and $n$.
\end{lemma}
\begin{proof}
  Fix $m$ and $n$. Consider the two sets of sizes $P_m$ and $P_n$ containing all prime balanced bracket expressions of size $m$ and $n$, respectively.
  By combining each pair of such expressions in a certain way, we will show how to obtain $P_m \cdot P_n$ distinct prime balanced bracket expressions of size $m+n$.
  
  Let $p_m$ and $p_n$ be two arbitrary prime balanced bracket expressions of respective sizes $m$ and $n$.
  We may assume that $m$ and $n$ are non-zero.
  Hence, we can assume that the expressions are of the form $p_m = p_m' \rangle$ and $p_n = \langle p_n'$, where $p_m'$ and $p_n'$ are the prefix and postfix, respectively, of $p_n$ and $p_m$, containing all but one letter.
  Here, we use the brackets $b_1 = \mlq\langle\mrq$ and $b_d = \mlq\rangle\mrq$, while leaving the remaining $d-2$ brackets unspecified.
  The expression corresponding to the pair $(p_m,p_n)$ is now defined as $p = p_m' \langle\rangle p_n'$.
  Clearly, in this way we obtain $P_m \cdot P_n$ distinct balanced bracket expressions of size $m+n$.
  It only remains to show that $p$ is prime.
  
  Consider thus a substring $c$ of $p$ that is a balanced bracket expression (of dimension $d$).
  Since by assumption $p_m$ and $p_n$ do not contain any such substrings, $c$ must contain the central ``$\langle\rangle$'' between $p_m'$ and $p_n'$ in $p$. 
  Fittingly, we write $c = c_m' \langle\rangle c_n'$, where $c_m'$ and $c_n'$ are a postfix and prefix, respectively, of $p_m'$ and $p_n'$.
  Furthermore, let $c_m = c_m'\rangle$ and $c_n = \langle c_n'$.
  The fact that $c_m'$ and $c_m$ are, respectively, a prefix and a postfix of a balanced bracket expression, easily implies that $c_m$ is a balanced bracket expression.
  By a symmetric argument, $c_n$ is also a balanced bracket expression.
  Since $p_m$ and $p_n$ are prime, it follows that $c_m = p_m$ and $c_n = p_n$, and thus $c = p$.
\end{proof}

Theorem \ref{thm:primeasym} from the introduction is a consequence of Lemma \ref{lem:supermulti}, Theorem \ref{thm:fekete}, and the fact that the radius of convergence of the formal power series $P(x)$ is equal to $1/\gamma_d$.
The latter is not hard to prove by using Lemma \ref{lem:relation} and by using that the radius of convergence of $C(x)$ is equal to $1/d^d$.

{
\renewcommand{\thetheorem}{\ref{thm:primeasym}}
\begin{theorem}
  For any dimension $d \geq 3$, the prime Catalan numbers satisfy
  \begin{equation*}
    \lim_{m \rightarrow \infty} \sqrt[m]{P_m} = \gamma_d\text{, } \; \text{ where } \gamma_d := \left(\frac{d}{C(1/d^d)}\right)^d\text{.} 
  \end{equation*}
\end{theorem}
}
\begin{proof}
  For any fixed dimension $d \geq 3$, let $R_C$ and $R_P$ be the radii of convergence of the power series $C(x)$ and $P(y)$, respectively.
  From the hook-length formula and Stirlings's approximation (see equations (\ref{eq:catalan_three_exact}) and (\ref{eq:catalan_three_asympt}) for the 3-dimensional case), it follows that
  \begin{equation}\label{eq:catalan_d_asympt}
    C_m \sim \frac{1! \cdot 2! \cdots (d-1)! \cdot \sqrt{d}}{\sqrt{2\pi}^{d-1}} m^{-(d^2-1)/2}d^{dm} \text{ \; (as $m \rightarrow \infty$),}
  \end{equation}
  and hence, by elementary analysis,
  \begin{equation}\label{eq:catalan_d_radius}
    R_C = 1/d^d\text{.}
  \end{equation}
  
  From equations (\ref{eq:catalan_d_asympt}) and (\ref{eq:catalan_d_radius}) we will deduce that $R_P = 1 / \gamma_d$.
  This will conclude the proof of the theorem because of Lemma \ref{lem:supermulti} and Theorem \ref{thm:fekete}.
  
  First, we show that $R_P \geq 1/\gamma_d$.
  Note that for positive $x < R_C$, the function $C(x)$ is continuous and strictly increasing, since all coefficients are positive.
  It follows that for every positive $y < 1/\gamma_d = R_C C(R_C)^d$ there exists a (unique) positive $x < R_C$ with $y = xC(x)^d$ and hence, with the help of Lemma \ref{lem:relation},
  \begin{equation*}
    P(y) = P(xC(x)^d) = C(x) < \infty\text{.}
  \end{equation*}
  
  Second, we show that $R_P \leq 1/\gamma_d$.
  Since the radius of convergence does not change under differentiation, it is sufficient to prove that the formal derivative of $P(y)$ of a certain order diverges at $y = 1/\gamma_d$.
  For that, we will use the following elementary observations.
  \begin{itemize}
    \item The $k$-th derivative\footnote{Be wary of the clash of notation here.} $C^{(k)}(R_C)$ remains convergent for $k < \frac{d^2-1}{2}-1$, but diverges for all $k \geq \frac{d^2-1}{2}-1$ (this follows from (\ref{eq:catalan_d_asympt}) by a comparison with hyperharmonic series).
    \item We have $C(x) > 0$ for all positive $x \leq R_C$ (simply because all coefficients are positive).
  \end{itemize}
  Let now $F(x) = xC(x)^d$, and consider the first derivative
  \begin{equation}\label{eq:derivative_f}
    F'(x) = dxC(x)^{d-1} \cdot C'(x) + C(x)^d\text{,}
  \end{equation}
  as well as the $k$-th derivative
  \begin{equation}\label{eq:derivative_f_k}
    F^{(k)}(x) = dxC(x)^{d-1} \cdot C^{(k)}(x) + \dots\text{,}
  \end{equation}
  where we have omitted all additive terms that contain only lower-order derivatives of $C(x)$.
  
  Starting from the equality given by Lemma \ref{lem:relation}, we similarly get
  \begin{equation}
    P'(F(x)) = \frac{C'(x)}{F'(x)}\text{,}
  \end{equation}
  as well as
  \begin{equation}\label{eq:derivative_p_k}
    P^{(k)}(F(x)) = \frac{C^{(k)}(x)}{F'(x)^k} - \frac{F^{(k)}(x)C'(x)}{F'(x)^{k+1}} + \dots\text{,}
  \end{equation}
  where again we have omitted additive terms that contain only lower-order derivatives of $C(x)$ and $F(x)$.
  By combining equation (\ref{eq:derivative_p_k}) with equations (\ref{eq:derivative_f_k}) and (\ref{eq:derivative_f}) we obtain the following.
  \begin{align}
    P^{(k)}(F(x))
  &= \frac{C^{(k)}(x)}{F'(x)^k} \Bigg(1 - \frac{dxC(x)^{d-1}C'(x)}{F'(x)}\Bigg) + \dots \\
  &= C^{(k)}(x) \cdot \frac{C(x)^d}{F'(x)^{k+1}} + \dots \label{eq:derivative_p_k_two}
  \end{align}
  Using our observations, for $k = \lceil\frac{d^2-1}{2}-1\rceil \geq 2$, as $x$ approaches $R_C$ from below, the right hand side of equation (\ref{eq:derivative_p_k_two}) diverges because $C^{(k)}(x)$ tends to infinity and because all omitted additive terms are bounded.
\end{proof}

As expected, the argument in the proof of Theorem~\ref{thm:primeasym} breaks down for the case $d=2$. Indeed, for $d=2$ we get $k = 1$, which means $F'(x)$ is no longer bounded and we cannot conclude that the right hand side of equation (\ref{eq:derivative_p_k_two}) tends to infinity. In fact, it does not diverge, since $P(y) = 1 + y$ and $P'(y) = 1$ for $d=2$.

\section{Acknowledgements}

None of the presented results would have been obtained without the help of the online encyclopedia of integer sequences \cite{OEIS}, which gave the crucial hint by recognizing the 3-dimensional Catalan numbers.
The author further would like to thank Emo Welzl, Jerri Nummenpalo, and Malte Milatz for interesting discussions on the subject.


\bibliography{bibliography}


\appendix
\section{Experiments}

In Tables \ref{tab:data_three}, \ref{tab:data_four}, \ref{tab:data_five} we present some experimental evidence for the asymptotic growth rate of the prime Catalan numbers of dimensions $d = 3,4,5$.

For $d = 3$ the corresponding approximations are defined as
\begin{align*}
  \tilde{C}^{(3)}_m := m^{-4} 3^{3m} && \tilde{P}^{(3)}_m := m^{-4} \gamma_3^m \text{,\;\;\; where $\gamma_3 = \frac{27}{(\frac{729\sqrt{3}}{40\pi} - 9)^3}$} \text{.}
\end{align*}

Note that, as can be expected from equation (\ref{eq:catalan_three_asympt}), the ratio $C^{(3)}_m / \tilde{C}^{(3)}_m$ approaches $\frac{\sqrt{3}}{\pi} \approx 0.55132$ as $m$ grows larger.
Also the ratio $P^{(3)}_m / \tilde{P}^{(3)}_m$ seems to converge, but we do not know the limit.

Similarly, for $d = 4$ we define
\begin{align*}
  \tilde{C}^{(4)}_m := m^{-7.5} 4^{4m} && \tilde{P}^{(4)}_m := m^{-7.5} \gamma_4^m \text{,\;\;\; where $\gamma_4 \approx 251.78874$} \text{.}
\end{align*}

Finally, for $d = 5$ we define
\begin{align*}
  \tilde{C}^{(5)}_m := m^{-12} 5^{5m} && \tilde{P}^{(5)}_m := m^{-12} \gamma_5^m \text{,\;\;\; where $\gamma_5 \approx 3119.93434$} \text{.}
\end{align*}

\begin{sidewaystable}
\begin{flushleft}
\begin{tabular}{rcccccc}
  \toprule
    $m$ & $C^{(3)}_m$ & $\tilde{C}^{(3)}_m$ & $C^{(3)}_m / \tilde{C}^{(3)}_m$ &
          $P^{(3)}_m$ & $\tilde{P}^{(3)}_m$ & $P^{(3)}_m / \tilde{P}^{(3)}_m$ \\
  \midrule
    1    & 1.00000e+00000 & 2.70000e+00001 & 0.03703 & 1.00000e+00000 & 2.34594e+00001 & 0.04262 \\
    2    & 5.00000e+00000 & 4.55625e+00001 & 0.10973 & 2.00000e+00000 & 3.43966e+00001 & 0.05814 \\
    4    & 4.62000e+00002 & 2.07594e+00003 & 0.22254 & 1.07000e+00002 & 1.18313e+00003 & 0.09043 \\
    8    & 2.33716e+00007 & 6.89525e+00007 & 0.33895 & 3.00297e+00006 & 2.23968e+00007 & 0.13408 \\
    16   & 5.21086e+00017 & 1.21713e+00018 & 0.42812 & 2.30416e+00016 & 1.28414e+00017 & 0.17943 \\
    32   & 2.94021e+00039 & 6.06792e+00039 & 0.48454 & 1.46103e+00037 & 6.75438e+00037 & 0.21630 \\
    64   & 1.24633e+00084 & 2.41302e+00084 & 0.51650 & 7.19612e+00079 & 2.98986e+00080 & 0.24068 \\
    128  & 3.25751e+00174 & 6.10550e+00174 & 0.53353 & 2.38674e+00166 & 9.37353e+00166 & 0.25462 \\
    256  & 3.39180e+00356 & 6.25408e+00356 & 0.54233 & 3.86180e+00340 & 1.47410e+00341 & 0.26197 \\
    512  & 5.74118e+00721 & 1.04994e+00722 & 0.54680 & 1.54989e+00690 & 5.83302e+00690 & 0.26571 \\
    1024 & 2.59965e+01453 & 4.73472e+01453 & 0.54906 & 3.91023e+01390 & 1.46132e+01391 & 0.26758 \\
    2048 & 8.47588e+02917 & 1.54052e+02918 & 0.55019 & 3.94041e+02792 & 1.46749e+02793 & 0.26851 \\
    4096 & 1.43714e+05848 & 2.60938e+05848 & 0.55076 & 6.36897e+05597 & 2.36783e+05598 & 0.26897 \\
    8192 & 6.60059e+11709 & 1.19783e+11710 & 0.55104 & 2.65530e+11209 & 9.86334e+11209 & 0.26920 \\
    16384& 2.22603e+23434 & 4.03861e+23434 & 0.55118 & 7.37501e+22433 & 2.73834e+22434 & 0.26932 \\
  \bottomrule
\end{tabular}
\end{flushleft}
\caption{Experimental data for 3-dimensional (prime) Catalan numbers.}\label{tab:data_three}
\end{sidewaystable}

\begin{sidewaystable}
\begin{flushleft}
\begin{tabular}{rcccccc}
  \toprule
    $m$ & $C^{(4)}_m$ & $\tilde{C}^{(4)}_m$ & $C^{(4)}_m / \tilde{C}^{(4)}_m$ &
          $P^{(4)}_m$ & $\tilde{P}^{(4)}_m$ & $P^{(4)}_m / \tilde{P}^{(4)}_m$ \\
  \midrule
    1    & 1.00000e+0000 & 2.56000e+0002 & 0.00390 & 1.00000e+0000 & 2.51788e+0002 &  0.00397 \\
    2    & 1.40000e+0001 & 3.62038e+0002 & 0.03866 & 1.00000e+0001 & 3.50225e+0002 &  0.02855 \\
    4    & 2.40240e+0004 & 1.31072e+0005 & 0.18328 & 1.67640e+0004 & 1.22657e+0005 &  0.13667 \\
    8    & 1.48987e+0012 & 3.10988e+0012 & 0.47907 & 1.05311e+0012 & 2.72342e+0012 &  0.38668 \\
    16   & 2.62708e+0029 & 3.16912e+0029 & 0.82896 & 1.70499e+0029 & 2.43041e+0029 &  0.70152 \\
    32   & 6.63875e+0065 & 5.95736e+0065 & 1.11437 & 3.36922e+0065 & 3.50378e+0065 &  0.96159 \\
    64   & 4.95456e+0140 & 3.81072e+0140 & 1.30016 & 1.49107e+0140 & 1.31817e+0140 &  1.13116 \\
    128  & 3.97058e+0292 & 2.82254e+0292 & 1.40674 & 4.14916e+0291 & 3.37732e+0291 &  1.22853 \\
    256  & 4.10340e+0598 & 2.80305e+0598 & 1.46390 & 5.14015e+0596 & 4.01325e+0596 &  1.28079 \\
    512  & 7.47391e+1212 & 5.00423e+1212 & 1.49351 & 1.34163e+1209 & 1.02581e+1209 &  1.30788 \\
    1024 & 4.35558e+2443 & 2.88718e+2443 & 1.50859 & 1.60345e+2436 & 1.21320e+2436 &  1.32167 \\
    2048 & 2.63772e+4907 & 1.73969e+4907 & 1.51619 & 4.08128e+4892 & 3.07180e+4892 &  1.32862 \\
  \bottomrule
\end{tabular}
\end{flushleft}
\caption{Experimental data for 4-dimensional (prime) Catalan numbers.}\label{tab:data_four}
\end{sidewaystable}

\begin{sidewaystable}
\begin{flushleft}
\begin{tabular}{rcccccc}
  \toprule
    $m$ & $C^{(5)}_m$ & $\tilde{C}^{(5)}_m$ & $C^{(5)}_m / \tilde{C}^{(5)}_m$ &
          $P^{(5)}_m$ & $\tilde{P}^{(5)}_m$ & $P^{(5)}_m / \tilde{P}^{(5)}_m$ \\
  \midrule
    1    & 1.00000e+0000 & 3.12500e+0003 &  0.00032 & 1.00000e+0000 & 3.11993e+0003 &   0.00032 \\
    2    & 4.20000e+0001 & 2.38418e+0003 &  0.01761 & 3.70000e+0001 & 2.37646e+0003 &   0.01556 \\
    4    & 1.66280e+0006 & 5.68434e+0006 &  0.29252 & 1.53347e+0006 & 5.64757e+0006 &   0.27152 \\
    8    & 2.31471e+0017 & 1.32348e+0017 &  1.74895 & 2.19820e+0017 & 1.30642e+0017 &   1.68261 \\
    16   & 1.46174e+0042 & 2.93873e+0041 &  4.97406 & 1.38606e+0042 & 2.86343e+0041 &   4.84055 \\
    32   & 5.23671e+0094 & 5.93472e+0093 &  8.82384 & 4.85822e+0094 & 5.63449e+0093 &   8.62228 \\
    64   & 1.18277e+0203 & 9.91383e+0201 & 11.93052 & 1.04361e+0203 & 8.93612e+0201 &  11.67859 \\
    128  & 1.57847e+0423 & 1.13313e+0422 & 13.93012 & 1.25644e+0423 & 9.20658e+0421 &  13.64719 \\
    256  & 9.13693e+0866 & 6.06353e+0865 & 15.06864 & 5.91142e+0866 & 4.00273e+0865 &  14.76846 \\
    512  & 1.11487e+1758 & 7.11166e+1756 & 15.67671 & 4.76248e+1757 & 3.09908e+1756 &  15.36741 \\
    1024 & 6.40765e+3543 & 4.00703e+3542 & 15.99101 & 1.19291e+3543 & 7.60931e+3541 &  15.67703 \\
    2048 & 8.41548e+7118 & 5.21056e+7117 & 16.15081 & 2.97532e+7117 & 1.87901e+7116 &  15.83445 \\
  \bottomrule
\end{tabular}
\end{flushleft}
\caption{Experimental data for 5-dimensional (prime) Catalan numbers.}\label{tab:data_five}
\end{sidewaystable}

\end{document}